\pdfoutput=1
\RequirePackage{ifpdf}
\ifpdf 
\documentclass[pdftex]{sigma}
\else
\documentclass{sigma}
\fi

\numberwithin{equation}{section}

\newtheorem{Theorem}{Theorem}[section]
\newtheorem{Corollary}[Theorem]{Corollary}
\newtheorem{Proposition}[Theorem]{Proposition}
 { \theoremstyle{definition}
\newtheorem{Definition}[Theorem]{Definition}
\newtheorem{Example}[Theorem]{Example}
\newtheorem{Remark}[Theorem]{Remark} }

\newcommand{\dis}{\displaystyle}

\begin{document}
\allowdisplaybreaks

\newcommand{\arXivNumber}{1901.01609}

\renewcommand{\PaperNumber}{048}

\FirstPageHeading

\ShortArticleName{Invariants in Separated Variables: Yang--Baxter, Entwining and Transfer Maps}

\ArticleName{Invariants in Separated Variables:\\ Yang--Baxter, Entwining and Transfer Maps}

\Author{Pavlos KASSOTAKIS}

\AuthorNameForHeading{P.~Kassotakis}

\Address{Department of Mathematics and Statistics, University of Cyprus,\\ P.O.~Box 20537, 1678 Nicosia, Cyprus}
\Email{\href{mailto:pavlos1978@gmail.com}{pavlos1978@gmail.com}, \href{mailto:pavlos1978@gmail.com}{pkasso01@ucy.ac.cy}}

\ArticleDates{Received January 16, 2019, in final form June 15, 2019; Published online June 25, 2019}

\Abstract{We present the explicit form of a family of Liouville integrable maps in 3~va\-riab\-les, the so-called {\it triad family of maps} and we propose a multi-field generalisation of the latter. We show that by imposing separability of variables to the invariants of this family of maps, the $H_{\rm I}$, $H_{\rm II}$ and~$H_{\rm III}^A$ Yang--Baxter maps in general position of singularities emerge. Two different methods to obtain entwining Yang--Baxter maps are also presented. The outcomes of the first method are entwining maps associated with the $H_{\rm I}$, $H_{\rm II}$ and~$H_{\rm III}^A$ Yang--Baxter maps, whereas by the second method we obtain non-periodic entwining maps associated with the whole $F$ and $H$-list of quadrirational Yang--Baxter maps. Finally, we show how the transfer maps associated with the $H$-list of Yang--Baxter maps can be considered as the $(k-1)$-iteration of some maps of simpler form. We refer to these maps as {\it extended transfer maps} and in turn they lead to $k$-point alternating recurrences which can be considered as alternating versions of some hierarchies of discrete Painlev\'e equations.}


\Classification{14E07; 14H70; 37K10}

\section{Introduction}
The quantum Yang--Baxter equation originates from the theory of exactly solvable models in statistical mechanics \cite{baxter-1982,yang-1967}. It reads
\begin{gather} \label{qyb}
R_{12} R_{13} R_{23} = R_{23} R_{13} R_{12},
\end{gather}
where $R\colon V\otimes V \mapsto V\otimes V$ a linear operator and $R_{lm}$, $l\neq m\in \{1,2,3\}$ the operators that acts as $R$ on the $l$-th and $m$-th factors of the tensor product
$V\otimes V \otimes V $. For the history of the latter and for the early developments on the theory see~\cite{Jimbo-1989}. Replacing the vector space $V$ with any set $X$ and the tensor product
 with the cartesian product, Drinfeld~\cite{drinfeld-1992} introduced the {\it set theoretical version} of~(\ref{qyb}). Solutions of the latter appeared under the name of {\it set theoretical solutions of the quantum Yang--Baxter equation}. The first instance of such solutions, appeared in~\cite{etingof-1999, sklyanin-1988}. The term {\it Yang--Baxter maps} was proposed by Veselov~\cite{Veselov:2003} as an alternative name to the Drinfeld's one. Early results on the context of Yang--Baxter maps were provided in \cite{adler-1993,KNY-2002A,Noumi1998}. In the recent years, many results arose in the interplay between studies on Yang--Baxter maps and the theory of discrete integrable systems \cite{Atkinson:2012i,AtkNie,Atkinson:2018,BAZHANOV2018509,Dimakis2018,Dimakis2018ii,Doliwa:2014,Grahovski:2016}.

In \cite{et-2003} it was considered a special type of set theoretical solutions of the quantum Yang--Baxter equation, the so called {\it non degenerate rational maps.} Nowadays, this type of solutions is referred to as {\it quadrirational Yang--Baxter maps}. Note that the notion of quadrirational maps, was extended in \cite{KaNie:2017} to the notion of {\it $2^n$-rational maps}, where highly symmetric higher dimensional maps were considered.
Under the assumption of quadrirationality and modulo conjugation (see Definition~\ref{def1}), in~\cite{ABSyb-2004,pap3-2010} a list of ten families of maps was obtained. Five of them were given in~\cite{ABSyb-2004}, which constitute the so-called $F$-list of quadrirational Yang--Baxter maps and five more in \cite{pap3-2010}, which constitute the so-called $H$-list of quadrirational Yang--Baxter maps. For their explicit form see Appendix~\ref{app1}. The Yang--Baxter maps of the $F$-list and the $H$-list can also be obtained from some of the integrable lattice equations in the classification scheme of \cite{ABS-2003}, by using the invariants of the generators of the Lie point symmetry group of the latter \cite{pap2-2006}. In the series of papers \cite{PKMN2,KaNie,PKMN3}, from the Yang--Baxter maps of the $F$-list and of the $H$-list, integrable lattice equations and correspondences (relations) were systematically constructed. Invariant, under the maps, functions where the variables appeared in separated form, played an important role to this construction. The cornerstone of this manuscript are invariant functions where the variables appear in separated form.

In \cite{Adler:2006}, it was introduced a family rational of maps in $3$ variables that preserves two rational functions the so-called {\it the triad map}. The triad map serves as a generalisation of the QRT map~\cite{QRT:1988} (cf.~\cite{Duistermaat:2010}). In Section~\ref{Section2} we present an explicit formula for Adler's triad map as well as we prove the Liouville integrability of the latter. We also propose an extension of the triad map in $k\geq 3$ number of variables. If one imposes separability to the variables of the invariants of the triad map, the $H_{\rm I}$, the $H_{\rm II}$ and the $H_{\rm III}^A$ Yang--Baxter maps in general positions of singularities, emerge. This is presented in Section~\ref{Section3} together with the explicit formulae for these maps.

In Section \ref{Section4}, we develop two methods to obtain non-equivalent {\it entwining maps} \cite{Kouloukas:2011}, i.e., maps $R$, $S$, $T$ that satisfy the relation
\begin{gather*}
 R_{12} S_{13} T_{23} = T_{23} S_{13} R_{12}.
\end{gather*}
 The first method gives us entwining maps associated with the $H_{\rm I}$, $H_{\rm II}$ and the $H_{\rm III}^A$ members of the $H$-list of Yang--Baxter maps. The second one produces entwining maps for the whole $F$-list and the $H$-list. In this manuscript we present the entwining maps associated with the $H$-list of quadrirational Yang--Baxter maps only.

In Section \ref{Section5}, we re-factorise the {\it transfer maps}~\cite{Veselov:2003} associated with the $H$-list of Yang--Baxter maps. We show that the transfer maps coincide with the $(k-1)$-iteration of some maps of simpler form that we refer to as {\it extended transfer maps}. Moreover, we show that the extended transfer maps, after an integration followed by a change of variables, are written as $k$-point recurrences, which some of them can be considered as alternating versions of discrete Painlev\'e hierarchies \cite{Joshi:1990,Hay:2007, Noumi1998}. In Section~\ref{Section6} we end this manuscript with some conclusions and perspectives.

\section{The Adler's triad family of maps}\label{Section2}
In \cite{Adler:2006}, Adler proposed the so-called {\it triad family of maps}. The triad map is a family of maps in 3 variables that consists of the composition of involutions which preserve two rational invariants of a specific form. In what follows we present the explicit form of the latter in terms of its invariants.

Consider the polynomials
\begin{gather*}
n^i=\sum_{j,k,l=0}^1\alpha^i_{j,k,l}x_1^{1-j}x_2^{1-k}x_3^{1-l},\qquad d^i=\sum_{j,k,l=0}^1\beta^i_{j,k,l}x_1^{1-j}x_2^{1-k}x_3^{1-l}, \qquad i=1,2,
\end{gather*}
where $x_1$, $x_2$, $x_3$ are considered as variables and $\alpha^i_{j,k,l}$, $\beta^i_{j,k,l}$ as parameters.
 We consider also $3$ maps $R_{ij}$, $i<j$, $i,j \in \{1,2,3\}$. These maps can be build out of the polynomials $n^i$, $d^i$ and they read $R_{ij}\colon (x_1,x_2,x_3)\mapsto (X_1(x_1,x_2,x_3),X_2(x_1,x_2,x_3),X_3(x_1,x_2,x_3))$, where
\begin{gather}
X_i=x_i-2\frac{\left|\begin{matrix}
{\mathbf D}_{x_i} n^1\cdot d^1 & {\mathbf D}_{x_i} n^2\cdot d^2\\
{\mathbf D}_{x_j} n^1\cdot d^1 & {\mathbf D}_{x_j} n^2\cdot d^2
\end{matrix}\right|}{\left|\begin{matrix}
{\mathbf D}_{x_i} n^1\cdot d^1 & {\mathbf D}_{x_i} n^2\cdot d^2\\
\partial_{x_i} {\mathbf D}_{x_j} n^1\cdot d^1+\partial_{x_j} {\mathbf D}_{x_i} n^1\cdot d^1 & \partial_{x_i} {\mathbf D}_{x_j} n^2\cdot d^2+
\partial_{x_j} {\mathbf D}_{x_i} n^2\cdot d^2
\end{matrix}\right|}, \nonumber\\
X_j=x_j+2\frac{\left|\begin{matrix}
{\mathbf D}_{x_i} n^1\cdot d^1 & {\mathbf D}_{x_i} n^2\cdot d^2\\
{\mathbf D}_{x_j} n^1\cdot d^1 & {\mathbf D}_{x_j} n^2\cdot d^2
\end{matrix}\right|}{\left|\begin{matrix}
{\mathbf D}_{x_j} n^1\cdot d^1 & {\mathbf D}_{x_j} n^2\cdot d^2\\
\partial_{x_i} {\mathbf D}_{x_j} n^1\cdot d^1+\partial_{x_j} {\mathbf D}_{x_i} n^1\cdot d^1 & \partial_{x_i} {\mathbf D}_{x_j} n^2\cdot d^2+
\partial_{x_j} {\mathbf D}_{x_i} n^2\cdot d^2
\end{matrix}\right|},\nonumber\\
X_k=x_k \qquad \mbox{for} \quad k\neq i,j, \label{Rij}
\end{gather} with $\partial_z$ we denote the partial derivative operator w.r.t.\ to $z$, i.e., $\partial_z h=\frac{\partial h}{\partial z}$. ${\mathbf D}_z$ is the Hirota's bilinear operator, i.e., ${\mathbf D}_z h\cdot k= ({\partial_z} h) k-h {\partial_z}k$.

\begin{Proposition} \label{prop1}The following holds:
\begin{enumerate}\itemsep=0pt
\item[$1.$] Mappings $R_{ij}$ depend on $32$ parameters $\alpha^i_{j,k,l}$, $\beta^i_{j,k,l}$, $i=1,2$, $j,k,l\in\{0,1\}$. Only $15$ of them are essential.
\item[$2.$] The functions $H_1=n^1/d^1$, $H_2=n^2/d^2$ are invariant under the action of $R_{ij}$, i.e., $H_l\circ R_{ij}=H_l$, $l=1,2$.
\item[$3.$] Mappings $R_{ij}$ are involutions, i.e., $R_{ij}^2={\rm id}$.
\item[$4.$] Mappings $R_{ij}$ are anti-measure preserving\footnote{A map $\phi\colon (x,y)\mapsto (X,Y)$ is called {\it measure preserving map} with density $m(x,y)$, if its Jacobian determinant $\frac{\partial (X,Y)}{\partial (x,y)}$ equals to $\frac{m(X,Y)}{m(x,y)}$. If the Jacobian determinant of the map $\phi$ equals to $-\frac{m(X,Y)}{m(x,y)}$, then the map $\phi$ is called {\it anti-measure preserving map} with density $m(x,y)$.} with densities $m_1=n^1 d^2$, $m_2=n^2 d^1$.
\item[$5.$] Mappings $R_{ij}$ satisfy the relation $R_{12} R_{13} R_{23}=R_{23} R_{13} R_{12}$.
\end{enumerate}
\end{Proposition}
\begin{proof}
1.~The invariants $H_1$, $H_2$ depend on $3$ variables and they include $32$ parameters. Acting with a different M\"obius transformation to each of the variables, 9~parameters can be removed. A~M\"obius transformation of an invariant remains an invariant, since we have $2$ invariants, $6$ more parameters can be removed. Finally, since any multiple of an invariant remains an invariant, $2$~more parameters can be removed. That leaves us with $32-9-6-2=15$ essential parameters for the invariants $H_1$, $H_2$ and hence for the maps $R_{ij}$.

2.~The functions $H_1=n^1/d^1$, $H_2=n^2/d^2$, reads
\begin{gather*}
H_1(x_1,x_2,x_3)=\frac{ax_1x_2+bx_1+cx_2+d}{a_1x_1x_2+b_1x_1+c_1x_2+d_1}, \\ H_2(x_1,x_2,x_3)=\frac{kx_1x_2+lx_1+mx_2+n}{k_1x_1x_2+l_1x_1+m_1x_2+n_1},
\end{gather*}
where $a,a_1,b,b_1,k,k_1,\ldots$ are linear functions of $x_3$ (note we have suppressed the dependency on $x_3$ of the functions $H_1, H_2$).
 From the set of equations
\begin{gather} \label{cons1}
H_1(X_1,X_2,x_3)=H_1(x_1,x_2,x_3),\qquad H_2(X_1,X_2,x_3)=H_2(x_1,x_2,x_3),
\end{gather}
by eliminating $X_2$ or by eliminating $X_1$ the resulting equations respectively factorize as
\begin{gather*}
(X_1-x_1)A=0,\qquad (X_2-x_2)B=0.
\end{gather*}
 The factor $A$ is linear in $X_1$ and the factor $B$ is linear in $X_2$. By solving these equations (we omit the trivial solution $X_1=x_1$, $X_2=x_2$) we obtain
 \begin{gather}
 X_1= \frac{\gamma_{13}^{34}x_2^2+\big(\gamma_{23}^{34}+\gamma_{14}^{34}\big)x_2+\gamma_{24}^{34}+\big(\gamma_{13}^{14}x_2^2+\big(\gamma_{13}^{24}+\gamma_{23}^{14}\big)x_2+
 \gamma_{23}^{24}\big)x_1}
 {\gamma_{13}^{23}x_2^2+\big(\gamma_{13}^{24}+\gamma_{14}^{23}\big)x_2+\gamma_{14}^{24}+\big(\gamma_{12}^{13}x_2^2+\big(\gamma_{12}^{23}+\gamma_{12}^{14}\big)x_2+
 \gamma_{12}^{24}\big)x_1} ,\nonumber\\
 X_2= \frac{\gamma_{12}^{24}x_1^2+\big(\gamma_{24}^{23}+\gamma_{14}^{24}\big)x_1+\gamma_{34}^{24}+\big(\gamma_{12}^{14}x_1^2+\big(\gamma_{12}^{34}+\gamma_{14}^{23}\big)x_1+
 \gamma_{34}^{23}\big)x_2}
 {\gamma_{23}^{12}x_1^2+\big(\gamma_{12}^{34}+\gamma_{23}^{14}\big)x_1+\gamma_{14}^{34}+\big(\gamma_{13}^{12}x_1^2+\big(\gamma_{23}^{13}+\gamma_{13}^{14}\big)x_1+
 \gamma_{13}^{34}\big)x_2} ,\label{Rij2}
\end{gather}
 where $ \gamma_{ij}^{kl}:=\left|\begin{smallmatrix}
 u_{ij}&u_{kl}\\
 v_{ij}&v_{kl} \end{smallmatrix}\right| $,
with $u_{ij}$ the determinants of a matrix generated by the $i^{\rm th}$ and $j^{\rm th}$ column of the matrix
\begin{gather*}
u=\left(\begin{matrix}
a&b&c&d\\
a_1&b_1&c_1&d_1
\end{matrix}\right)
\end{gather*}
and $v_{kl}$ the determinants of a matrix generated by the $k^{\rm th}$ and $l^{\rm th}$ column of the matrix
\begin{gather*}
v=\left(\begin{matrix}
k&l&m&n\\
k_1&l_1&m_1&n_1
\end{matrix}\right).
\end{gather*}
Now it is a matter of long and tedious calculation to prove that the map $\phi\colon (x_1,x_2,x_3)\mapsto(X_1,X_2,x_3)$, where $X_1,X_2$
are given by~(\ref{Rij2}) coincides with the map $R_{12}$ of~(\ref{Rij}). Similarly we can work on~$R_{13}$ and~$R_{23}$.

3.~Since the map $R_{12}\colon (x_1,x_2,x_3)\mapsto(X_1,X_2,x_3)$ satisfies~(\ref{cons1}), the proof of involutivity follows.

4.~It is enough to prove that the map $R_{12}$ anti-preserves the measure with density $m_1=n^1d^2$, i.e., the Jacobian determinant
\begin{gather*} \frac{\partial{(X_1,X_2)}}{\partial{(x_1,x_2)}}:=\left|\begin{matrix}
 \dfrac{\partial X_1}{\partial x_1}& \dfrac{\partial X_1}{\partial x_2}\vspace{1mm}\\
 \dfrac{\partial X_2}{\partial x_1}& \dfrac{\partial X_2}{\partial x_2}
 \end{matrix}\right|\end{gather*} equals
 \begin{gather*} 
 \frac{\partial{(X_1,X_2)}}{\partial{(x_1,x_2)}}=-\frac{n^1(X_1,X_2,x_3)d^2(X_1,X_2,x_3)}{n^1(x_1,x_2,x_3)d^2(x_1,x_2,x_3)}.
 \end{gather*}
 Since the functions $H_i=n^i/d^i$, $i=1,2$ are invariant under the action of the map $R_{12}$, it holds
 \begin{gather}
 n^1(X_1,X_2,x_3)=\kappa(x_1,x_2,x_3)n^1(x_1,x_2,x_3),\nonumber\\
 d^1(X_1,X_2,x_3)=\kappa(x_1,x_2,x_3)d^1(x_1,x_2,x_3),\nonumber\\
 n^2(X_1,X_2,x_3)=\lambda(x_1,x_2,x_3)n^2(x_1,x_2,x_3),\nonumber\\
 d^2(X_1,X_2,x_3)=\lambda(x_1,x_2,x_3)d^2(x_1,x_2,x_3),\label{der1}
 \end{gather}
 where $\kappa$, $\lambda$ are rational functions of $x_1$, $x_2$, $x_3$. So,
 \begin{gather} \label{mes1}
 \frac{n^1(X_1,X_2,x_3)d^2(X_1,X_2,x_3)}{n^1(x_1,x_2,x_3)d^2(x_1,x_2,x_3)}=\kappa(x_1,x_2,x_3) \lambda(x_1,x_2,x_3).
 \end{gather}
 We differentiate equations~(\ref{der1}) with respect to $x_1$ and we eliminate $\frac{\partial \kappa(x_1,x_2,x_3)}{\partial x_1}$ and $\frac{\partial \lambda(x_1,x_2,x_3)}{\partial x_1}$ to obtain
 \begin{gather}
\frac{1}{n^1}\left(\frac{\partial \tilde n^1}{\partial x_1}-\kappa\frac{\partial n^1}{\partial x_1}\right)= \frac{1}{d^1}\left(\frac{\partial \tilde d^1}{\partial x_1}-\kappa\frac{\partial d^1}{\partial x_1}\right) ,\nonumber\\
 \frac{1}{n^2}\left(\frac{\partial \tilde n^2}{\partial x_1}-\lambda\frac{\partial n^2}{\partial x_1}\right)= \frac{1}{d^2}\left(\frac{\partial \tilde d^2}{\partial x_1}-\lambda\frac{\partial d^2}{\partial x_1}\right),\label{p41}
 \end{gather}
 here we have suppressed the dependency of $\kappa$, $\lambda$, $n^i$, $d^i $ on $x_1$, $x_2$, $x_3$. By $\tilde n^i$ we denote $\tilde n^i:=n^i(X_1,X_2,x_3)$, $i=1,2$, and similarly for $\tilde d^i$.
 Also if we differentiate the equations~(\ref{der1}) with respect to~$x_2$ and eliminate $\frac{\partial \kappa}{\partial x_2}$ and $\frac{\partial \lambda}{\partial x_2}$ we obtain
 \begin{gather}
\frac{1}{n^1}\left(\frac{\partial \tilde n^1}{\partial x_2}-\kappa\frac{\partial n^1}{\partial x_2}\right)= \frac{1}{d^1}\left(\frac{\partial \tilde d^1}{\partial x_2}-\kappa\frac{\partial d^1}{\partial x_2}\right) , \nonumber\\
\frac{1}{n^2}\left(\frac{\partial \tilde n^2}{\partial x_2}-\lambda\frac{\partial n^2}{\partial x_2}\right)= \frac{1}{d^2}\left(\frac{\partial \tilde d^2}{\partial x_2}-\lambda\frac{\partial d^2}{\partial x_2}\right).\label{p42}
 \end{gather}
 Due to the form of $n^i$, $d^i$, $i=1,2$, equations (\ref{p41}), (\ref{p42}) are linear in $\frac{\partial X_1}{\partial x_i}$, $\frac{\partial X_2}{\partial x_i}$, $i=1,2$. Hence we obtain $\frac{\partial X_1}{\partial x_i}$, $\frac{\partial X_2}{\partial x_i}$, $i=1,2, $ in terms of $X_1$, $X_2$, $x_1$, $x_2$, $x_3$, $\kappa$, $\lambda$ and by using~(\ref{Rij2}), the Jacobian determinant reads
 $ \frac{\partial{(X_1,X_2)}}{\partial{(x_1,x_2)}}=-\kappa\lambda $. Using~(\ref{mes1}) we have
 \begin{gather*}
 \frac{\partial{(X_1,X_2)}}{\partial{(x_1,x_2)}}=-\kappa\lambda=-\frac{\tilde n^1\tilde d^2}{n^1d^2},
 \end{gather*}
 that completes the proof. Note that the same holds true for the remaining maps $R_{ij}$.

5.~In \cite{Adler:2006} Adler presented a computational proof based on the fact that the maps $R_{ij}$ map points that lie on the invariant curve
 \begin{gather}\label{adl-prf}
 n^1(x_1,x_2,x_3)-C_1d^1(x_1,x_2,x_3)=0,\qquad n^2(x_1,x_2,x_3)-C_2d^2(x_1,x_2,x_3)=0,
 \end{gather}
 that is the intersection of two surfaces of the form $A\colon N(x_1,x_2,x_3)=0$, where $N$ is polynomial with degree at most one on each variable $x_1$, $x_2$ and $x_3$. In~\cite{Adler:2006}, it was proven that any surface of the form $A$ that passes through the following five points
 \begin{gather*} 
 \big(\hat{x_1},\tilde{x}_2,\hat{\tilde{x}}_3\big) \xleftarrow{R_{13}} \big(x_1,\tilde{x}_2,\tilde{x}_3\big) \xleftarrow{R_{23}} (x_1,x_2,x_3)\xrightarrow{R_{12}} (\bar x_1,{\bar x_2},x_3) \xrightarrow{R_{13}} (\hat{\bar{x}}_1,\bar{x}_2,\hat{x}_3)
 \end{gather*}
 passes as well through the point $\big(\hat{\bar{x}}_1,Y,\tilde{\hat{x}}_3\big)$, that is the point of intersection of the straight line $L\colon (X,Z)=\big(\hat{\bar{x}}_1,\tilde{\hat{x}}_3\big)$ and the surface $A$, i.e., $L\cap A=\big(\hat{\bar{x}}_1,Y,\tilde{\hat{x}}_3\big)$. Since the invariant curve~(\ref{adl-prf}) is the intersection of two surfaces of the form~$A$, it also passes through the point $\big(\hat{\bar{x}}_1,Y,\tilde{\hat{x}}_3\big)$ and there is $\tilde{\bar{x}}_2=Y$. So the values of $\tilde{\bar{x}}_2$ obtained in two different ways coincide and this is sufficient for the proof.

Alternatively, one can show by direct computation that the maps $T_1=R_{13}R_{12}$ and $T_2=R_{12}R_{23}$, commute, i.e., $T_1T_2=T_2T_1$. So there is
\begin{gather*}
R_{13}R_{23}=R_{12}R_{23}R_{13}R_{12}
\end{gather*}
and due to the fact that the maps $R_{ij}$ are involutions, $R_{ij}^2={\rm id}$, from the equation above we obtain
\begin{gather*}
R_{12}R_{13}R_{23}=R_{23}R_{13}R_{12}.\tag*{\qed}
\end{gather*}\renewcommand{\qed}{}
\end{proof}

Among all the maps that can be constructed by the involutions $R_{ij}$, the following maps
\begin{gather*}
T_1=R_{13}R_{12},\qquad T_2=R_{12}R_{23}, \qquad T_3=R_{23}R_{13}
\end{gather*}
are of special interest since they are not periodic and moreover they satisfy~\cite{Adler:2006}
\begin{gather*}
T_1T_2T_3={\rm id}, \qquad T_iT_j=T_jT_i, \qquad i,j\in \{1,2,3\}.
\end{gather*}
\begin{Proposition}
For the maps $T_i$, $i=1,2,3$ it holds:
\begin{enumerate}\itemsep=0pt
\item[$1)$] they preserve the functions $H_1$, $H_2$,
\item[$2)$] they are measure-preserving with densities $m_1$, $m_2$,
\item[$3)$] they preserve the following degenerate Poisson tensors,
\begin{gather*}
\Omega_i^j=m_j\left( \frac{\partial H_i}{\partial x_3} \frac{\partial }{\partial x_1}\wedge\frac{\partial }{\partial x_2} -\frac{\partial H_i}{\partial x_2} \frac{\partial }{\partial x_1}\wedge\frac{\partial }{\partial x_3}+\frac{\partial H_i}{\partial x_1} \frac{\partial }{\partial x_2}\wedge\frac{\partial }{\partial x_3}\right), \qquad i,j\in\{1,2\},
\end{gather*}
where it holds
\begin{gather*}
0=\Omega_1^j \nabla H_1, \qquad \Omega_1^j \nabla H_2=-\Omega_2^j \nabla H_1, \qquad \Omega_2^j \nabla H_2=0, \qquad j=1,2,
\end{gather*}
\item[$4)$] they are Liouville integrable maps.
\end{enumerate}
\end{Proposition}
\begin{proof}
The statements $(1)$, $(2)$ follows from Proposition~\ref{prop1}. To prove the statement $(3)$, $(4)$, first note that since the maps $T_i$ are measure preserving, they preserve the following polyvector fields
\begin{gather*}
V^j=m_j\frac{\partial }{\partial x_1}\wedge\frac{\partial }{\partial x_2}\wedge\frac{\partial }{\partial x_3}.
\end{gather*}
Hence, the contractions $V^j\rfloor d H_i$, $i,j\in\{1,2\}$ (see~\cite{Kassotakis:2006, Haggar:1996}) are degenerate Poisson tensors. Namely,
\begin{align*}
\Omega_i^j& =\left(m_j\frac{\partial }{\partial x_1}\wedge\frac{\partial }{\partial x_2}\wedge\frac{\partial }{\partial x_3}\right)\rfloor d H_i\\
& =m_j\left( \frac{\partial H_i}{\partial x_3} \frac{\partial }{\partial x_1}\wedge\frac{\partial }{\partial x_2} -\frac{\partial H_i}{\partial x_2} \frac{\partial }{\partial x_1}\wedge\frac{\partial }{\partial x_3}+\frac{\partial H_i}{\partial x_1} \frac{\partial }{\partial x_2}\wedge\frac{\partial }{\partial x_3}\right),
\end{align*}
where $i,j\in\{1,2\}$.

$(5)$ The maps $T_i$ preserve the Poisson tensors $\Omega_i^j$ and the $2$ invariants $H_1$, $H_2$, so they are Liouville integrable maps~\cite{Santini:1991,Maeda:1987,Veselov:1991}.
\end{proof}

Note that on the level surfaces $H_2(x_1,x_2,x_3)=c$, maps $T_1$, $T_2$, $T_3$ reduce to pair-wise commuting maps on the plane which preserve the function
$\hat H_1(x_1,x_2;c)$. One of these reduced maps is the associated with the invariant $\hat H_1(x_1,x_2;c)$ QRT map. Examples of commuting maps with specific members of the QRT family of maps were also constructed in \cite{Kassotakis:2013}.

The involution $R_{12}$ under the reduction $x_2=x_1$, $H_2=H_1=H$, so $H=\frac{n}{d}=\frac{ax_1^2+bx_1+c}{kx_1^2+lx_1+m}$, reads
\begin{gather*}
R_{12}\colon \ (x_1,x_3)\mapsto \left(x_1-2\frac{D_{x_1} n\cdot d}{\partial_{x_1} D_{x_1} n\cdot d},x_3\right),
\end{gather*}
that coincides with the QRT involution $i_x$ that preserves the invariant~$H$. This formulae for the QRT involution $i_x$ was firstly given in~\cite{Joshi:2018}, where an elegant presentation of the QRT map was considered.

\subsection{A generalisation of the triad family of maps} \label{subsection2.1}

Following the same generalisation procedures introduced for the QRT family of maps \cite{Capel:2001,Kassotakis:2006,Iatrou:2003,QRT:2006,Tsuda:2004}, the triad family of maps can be generalised in similar manners. Here, in order to generalise the triad family of maps, we mimic the generalisation of the QRT family of maps introduced in~\cite{Tsuda:2004}.

Consider the following polynomials
\begin{gather}
 n^i=\sum_{j_1,j_2,\ldots,j_k=0}^1\alpha^i_{j_1,j_2,\ldots,j_k}x_1^{1-j_1}x_2^{1-j_2}\cdots x_k^{1-j_k},\nonumber\\
d^i=\sum_{j_1,j_2,\ldots,j_k=0}^1\beta^i_{j_1,j_2,\ldots,j_k}x_1^{1-j_1}x_2^{1-j_2}\cdots x_k^{1-j_k}, \qquad
i=1,2 k\geq 3,\label{def-pol}
\end{gather}
where $x_1,x_2,\ldots, x_k$ are considered as variables and $\alpha^i_{j_1,j_2,\ldots,j_k}$, $\beta^i_{j_1,j_2,\ldots,j_k}$ as parameters.
 We consider the $\binom{k}{2}$ maps $R_{ij}$, $i<j$, $i,j \in \{1,2,\ldots,k\}$. These maps can be build out of the polynomials $n^i$, $d^i$ and they read:
$R_{ij}\colon (x_1,x_2,\ldots, x_k)\mapsto (X_1,X_2,\ldots, X_k)$, where $X_l=x_l$ $\forall\, l\neq i,j$ and $X_i$, $X_j$ are given by the formulae~(\ref{Rij}), where $n^i$, $d^i$, $i=1,2$ are given by~(\ref{def-pol}).

Proposition \ref{prop1} is straight forward extended to the $k$-variables case.
\begin{Proposition}The following holds:
\begin{enumerate}\itemsep=0pt
\item[$1.$] Mappings $R_{ij}$ depend on $4\cdot 2^k$ parameters $\alpha^i_{j_1,j_2,\ldots,j_k}$, $\beta^i_{j_1,j_2,\ldots,j_k}$, $i=1,2$, $j_1,j_2,\ldots,j_k\in\{0,1\}$.
 Only $4\cdot2^k-3k-8$ of them are essential.
\item[$2.$] The functions $H_1=n^1/d^1$, $H_2=n^2/d^2$ are invariant under the action of $R_{ij}$, i.e., $H_l \circ R_{ij}=H_l$, $l=1,2$.
\item[$3.$] Mappings $R_{ij}$ are involutions, i.e., $R_{ij}^2={\rm id}$.
\item[$4.$] Mappings $R_{ij}$ are anti-measure preserving with densities $m_1=n^1 d^2$, $m_2=n^2 d^1$.
\item[$5.$] Mappings $R_{mn}$, $ m<n$, $m,n \in \{1,2,\ldots,k\}$ satisfy the relations $R_{ij} R_{il} R_{jl}=R_{jl} R_{il} R_{ij}$.
\end{enumerate}
\end{Proposition}
\begin{proof} 1.~ The invariants $H_1$, $H_2$ depend on $k\geq 3$ variables and they include $4\cdot 2^k$ parameters. Acting with a different M\"obius transformation to each of the variables, $3k$ parameters can be removed. A M\"obius transformation of an invariant remains an invariant, since we have $2$ invariants, $6$ more parameters can be removed. Finally, since any multiple of an invariant remains an invariant, $2$ more parameters can be removed. That leaves us with $4\cdot 2^k-3k-6-2=4\cdot 2^k-3k-8$ essential parameters for the invariants $H_1$, $H_2$ and hence for the maps~$R_{ij}$.

 The proof of the remaining statements of this Proposition follows directly from the fact that for any 3 indices $p<q< r\in \{1,2,\ldots,k\}$, the maps $R_{pq}$, $R_{pr}$ and $R_{qr}$, coincide with the maps $R_{12}$, $R_{13}$ and $R_{23}$ respectively of Proposition~\ref{prop1}.
\end{proof}

We take a stand here to comment that for $k=3$ the construction above coincides with the Adler's triad family of maps hence we have Liouville integrability. For $k>3$ we have a~generalisation of the latter and since always we will have maps in $k$ variables with $2$ invariants, Liouville integrability is not expected for generic choice of the parameters $\alpha^i_{j_1,j_2,\ldots,j_k}$, $\beta^i_{j_1,j_2,\ldots,j_k}$. For a specific but quite general choice of the parameters though, one can associate a Lax pair to these maps and recover the additional integrals which are required for the Liouville integrability to emerge.

We also have to note that the case $k=4$ was firstly introduced in~\cite{Kassotakis:phd}. Although for $k=4$ we have mappings in $4$ variables with $2$ invariants, Liouville integrability is not apparent unless we specify the parameters. A specific choice of the parameters which leads to integrability is presented to the following example.

\begin{Example}[the Adler--Yamilov map \cite{Adler:1994}]
Consider the following special form of the func\-tions~$n^i$,~$d^i$
\begin{gather*}
d^1=d^2=1,\qquad n^1=x_1x_2+x_3x_4,\qquad n^2=x_1x_2x_3x_4+x_1x_4+x_2x_3+a x_1x_2+b x_3x_4.
\end{gather*}
Then the functions $H_i=n^i/d^i$, $i=1,2$ are preserved by construction by the maps $R_{ij}$ as well as by the following elementary involutions
\begin{gather*}
i\colon \ (x_1,x_2,x_3,x_4)\mapsto (x_2,x_1,x_4,x_3),\qquad \phi\colon \ (x_1,x_2,x_3,x_4)\mapsto (x_1x_2/x_3,x_3,x_2,x_3x_4/x_2).
\end{gather*}
The Adler--Yamilov map ($\xi$) is considered by the following composition
\begin{gather*}
\xi:=R_{14} \phi i\colon \ (x_1,x_2,x_3,x_4)\mapsto \left(x_3-\frac{(a-b)x_1}{1+x_1x_4},x_4,x_1,x_2+\frac{(a-b)x_4}{1+x_1x_4}\right).
\end{gather*}
The Adler--Yamilov map is Liouville integrable since it preserves, and the invariants $H_1$, $H_2$ are in involution with respect to the canonical Poisson bracket.
For further discussions on the Adler--Yamilov map see~\cite{Kassotakis:2013, Rizos:2013}.
\end{Example}

\section{Invariants in separated variables and Yang--Baxter maps}\label{Section3}

Mappings $R_{mn}$, $m<n \in \{1,2,\ldots, k\}$, presented in Section~\ref{subsection2.1}, satisfy the identities $R_{ij} R_{il} R_{jl}=R_{jl} R_{il} R_{ij}$, nevertheless as they stand they are not Yang--Baxter. Take for example the map $R_{12}\colon (x_1,x_2,x_3,\ldots, x_k)\mapsto (X_1,X_2,x_3,\ldots, x_k)$. The formulae for $X_1$ is fraction linear in $x_1$ with coefficients that depend on all the remaining variables and $X_2$ is fraction linear in $x_2$ with coefficients that depend on all the remaining variables. In order for $R_{12}$ to be a~Yang--Baxter map the coefficients of $x_1$ in the formulae of $X_1$ should depend only on $x_2$ and the coefficients of $x_2$ in the formulae of $X_2$ should depend only on $x_1$. This ``separability" requirement can be easily achieved by requiring separability of variables on the level of the invariants of the map~$R_{12}$. We have two invariants $H_1=n^1/d^1$, $H_2=n^2/d^2$, so we can have three different kinds of separability. (I)~Both $H_1$ and $H_2$ to be multiplicative separable on the variables~$x_1$ and~$x_2$. (II)~$H_1$ to be multiplicative and $H_2$ to be additive separable and finally (III)~both $H_1$ and $H_2$ to be additive separable on the variables~$x_1$ and~$x_2$. In what follows we explicitly present these three different kinds of separability in all variables of the invariants $H_1$ and $H_2$.
\begin{enumerate}\itemsep=0pt
\item[(I)] Multiplicative/multiplicative separability of variables:
\begin{gather} \label{multi/multi}
H_1=\prod_{i=1}^k\frac{a_i-b_ix_i}{c_i-d_ix_i},\qquad H_2=\prod_{i=1}^k\frac{A_i-B_ix_i}{C_i-D_ix_i}.
\end{gather}
\item[(II)] Multiplicative/additive separability of variables:
\begin{gather} \label{multi/add}
H_1=\prod_{i=1}^k\frac{a_i-b_ix_i}{c_i-d_ix_i},\qquad H_2=\sum_{i=1}^k\frac{A_i-B_ix_i}{C_i-D_ix_i}.
\end{gather}
\item[(III)] Additive/additive separability of variables:
\begin{gather} \label{add/add}
H_1=\sum_{i=1}^k\frac{a_i-b_ix_i}{c_i-d_ix_i},\qquad H_2=\sum_{i=1}^k\frac{A_i-B_ix_i}{C_i-D_ix_i}.
\end{gather}
\end{enumerate}
In the formulas above, $a_i$, $b_i$, $c_i$, $d_i$, $A_i$, $B_i$, $C_i$, $D_i$, $i=1,\ldots, k$ are parameters, $8k$ in total. In all three cases above, the number of essential parameters is $3k-6$. This argument can be proven by the following reasoning. Since the invariants $H_1$, $H_2$ depends on $k$ variables, by a~M\"{o}bius transformation on each of the $k$ variables $3k$ parameters can be removed. Also any M\"{o}bius transformation of an invariant remains an invariant so since we have two invariants $2\times3$ more parameters can be removed. Finally, for each one of the $2k$ functions $\frac{a_i-b_ix_i}{c_i-d_ix_i}$, $\frac{A_i-B_ix_i}{C_i-D_ix_i}$, $i=1,\ldots, k$, one non-zero parameter can be absorbed simply by dividing with it (and reparametrise), so $2k$ more parameters can be removed. In total we have $8k-3k-2\times 3-2k=3k-6$ essential parameters.

\subsection{Multiplicative/multiplicative separability of variables} \label{subsection3.1}
Let us first introduce some definitions.
\begin{Definition} \label{def1}
The maps $R,\tilde R \colon \mathbb{CP}^1\times \mathbb{CP}^1\mapsto \mathbb{CP}^1\times \mathbb{CP}^1$ are ${\text{(M\"{o}b)}}^2$ {\it equivalent} if there exists bijections $\phi ,\psi \colon \mathbb{CP}^1\mapsto \mathbb{CP}^1$ such that the following conjugation relation holds
\begin{gather*}
\tilde R=\phi^{-1}\times \psi^{-1}R\phi\times\psi.
\end{gather*}
\end{Definition}
\begin{Definition} \label{def2}
The map $R\colon \mathbb{CP}^1\times \mathbb{CP}^1\ni (u,v)\mapsto (U,V)\in \mathbb{CP}^1\times \mathbb{CP}^1$, where
\begin{gather*}
U=\frac{a_1+a_2u}{a_3+a_4u},\qquad V=\frac{b_1+b_2v}{b_3+b_4v},
\end{gather*}
with $a_i$, $b_i$, $i=1,\ldots,4$ known polynomials of $v$ and $u$ respectively, will be said to be of {\it subclass $[\gamma:\delta]$,} if the highest degree that appears in the polynomials $a_i$ is $\gamma$ and the higher degree that appears in the polynomials~$b_i$ is~$\delta$.
\end{Definition}
Clearly, maps that belong to different subclasses are not ${\text{(M\"{o}b)}}^2$ equivalent.
\begin{Proposition} \label{prop2}
Consider the multiplicative/multiplicative separability of variables of the invariants $H_1$ and $H_2$ $($see \eqref{multi/multi}$)$. Consider also the following sets of parameters
\begin{gather*}
{\bf p}_{ij}:={\bf p}_{i}\cup {\bf p}_{j} \qquad \text{where} \qquad {\bf p}_{i}:= \{a_i,b_i,c_i,d_i,A_i,B_i,C_i,D_i \}, \qquad i<j \in \{1,2,\ldots, k\}
\end{gather*}
and the functions
\begin{gather*}
f_i:=\frac{a_i-b_ix_i}{c_i-d_ix_i}, \qquad g_i:=\frac{A_i-B_ix_i}{C_i-D_ix_i}, \qquad i=1,\ldots, k.
\end{gather*}
The following holds:
\begin{enumerate}\itemsep=0pt
\item[$1.$] The invariants $H_1=\prod\limits_{i=1}^k f_i$, $H_2=\prod\limits_{i=1}^k g_i$ depend on $8k$ parameters. Only $3k-6$ of them are essential.
\item[$2.$] Mappings $R_{ij}$ explicitly read
\begin{gather*}
R_{ij}\colon \ (x_1,x_2,\ldots, x_k)\mapsto (X_1,X_2,\ldots, X_k),
\end{gather*}
where $X_l=x_l$ $\forall\, l\neq i,j$ and $X_i$, $X_j$ are given by the formulae
 \begin{gather*}
 X_i= x_i-2\frac{\left|\begin{matrix}
 f_i'f_j&f_if_j'\\
 g_i'g_j&g_ig_j'
 \end{matrix}\right|}{g_i'g_j\left(\dfrac{f_i'}{f_j'}\left|\begin{matrix}
 f_j&f_j'\\
 f_j'&f_j''\end{matrix}\right|+\dfrac{f_j'}{f_i'}\left|\begin{matrix}
 f_i&f_i'\\
 f_i'&f_i''\end{matrix}\right| \right)-f_i'f_j\left(\dfrac{g_i'}{g_j'}
 \left|\begin{matrix}
 g_j&g_j'\\
 g_j'&g_j''\end{matrix}\right|+
 \dfrac{g_j'}{g_i'}
 \left|\begin{matrix}
 g_i&g_i'\\
 g_i'&g_i''\end{matrix}\right| \right)},\\
 X_j=x_j+2\frac{\left|\begin{matrix}
 f_i'f_j&f_if_j'\\
 g_i'g_j&g_ig_j'
 \end{matrix}\right|}{g_j'g_i\left(\dfrac{f_i'}{f_j'}\left|\begin{matrix}
 f_j&f_j'\\
 f_j'&f_j''\end{matrix}\right|+\dfrac{f_j'}{f_i'}\left|\begin{matrix}
 f_i&f_i'\\
 f_i'&f_i''\end{matrix}\right| \right)-f_j'f_i\left(\dfrac{g_i'}{g_j'}
 \left|\begin{matrix}
 g_j&g_j'\\
 g_j'&g_j''\end{matrix}\right|+
 \dfrac{g_j'}{g_i'}
 \left|\begin{matrix}
 g_i&g_i'\\
 g_i'&g_i''\end{matrix}\right| \right)},
\end{gather*}
where $f_l'\equiv \frac{\partial f_l}{\partial x_l}$, $g_l'\equiv \frac{\partial g_l}{\partial x_l}$, $g_l''\equiv \frac{\partial^2 g_l}{\partial x_l^2}$, etc. Note that in the expressions of $X_i$, $X_j$ appears only the coordinates $x_i$, $x_j$ and the parameters ${\bf p}_{ij}$. From further on we denote the maps~$R_{ij}$ as~$R_{ij}^{{\bf p}_{ij}}$, in order to stress this separability feature.
\item[$3.$] Mappings $R_{ij}^{{\bf p}_{ij}}$ are anti-measure preserving with densities $m_1=n^1 d^2$, $m_2=n^2 d^1$, where $n^i$, $d^i$ the numerators and the denominators respectively, of the invariants $H_i$, $i=1,2$.
\item[$4.$] Mappings $R_{ij}^{{\bf p}_{ij}}$ satisfy the Yang--Baxter identity
\begin{gather*}
R_{ij}^{{\bf p}_{ij}} R_{ik}^{{\bf p}_{ik}} R_{jk}^{{\bf p}_{jk}}=R_{jk}^{{\bf p}_{jk}} R_{ij}^{{\bf p}_{ij}} R_{ij}^{{\bf p}_{ij}}.
\end{gather*}
\item[$5.$] Mappings $R_{ij}^{{\bf p}_{ij}}$ are involutions with the sets of singularities
\begin{gather*}
\Sigma_{ij}=\big\{P^1_{ij},P^2_{ij},P^3_{ij},P^4_{ij} \big\}=\left\{\left(\frac{a_i}{b_i},\frac{c_j}{d_j}\right),\left(\frac{c_i}{d_i},\frac{a_j}{b_j}\right),\left(\frac{A_i}{B_i},\frac{C_j}{D_j}\right),
\left(\frac{C_i}{D_i},\frac{A_j}{B_j}\right)\right\},
\end{gather*}
 and the sets of fixed points \begin{gather*}
 \Phi_{ij}=\big\{Q^1_{ij},Q^2_{ij},Q^3_{ij},Q^4_{ij} \big\}=\left\{\left(\frac{a_i}{b_i},\frac{a_j}{b_j}\right),\left(\frac{c_i}{d_i},\frac{c_j}{d_j}\right),\left(\frac{A_i}{B_i},\frac{A_j}{B_j}\right),
 \left(\frac{C_i}{D_i},\frac{C_j}{D_j}\right)\right\},
\end{gather*}
where in the formulae for $P_{ij}^m$ and $Q_{ij}^m$, $m=1,\ldots, 4$, we have suppressed the dependency on the remaining variables. For example, with
$P_{ij}^1=\big(\frac{a_i}{b_i},\frac{c_j}{d_j}\big)$ we denote $\big(x_1,\ldots,x_{i-1},\frac{a_i}{b_i},x_{i+1},\allowbreak \ldots,x_{j-1},\frac{c_j}{d_j},x_{j+1},\ldots, x_k\big)$ and similarly for the remaining $P_{ij}^m$ and~$Q_{ij}^m$.

\item[$6.$] Each one of the maps $R_{ij}^{{\bf p}_{ij}}$ is ${\text{\rm (M\"{o}b)}}^2$ equivalent to the $H_{\rm I}$ Yang--Baxter map.
\end{enumerate}
\end{Proposition}
\begin{proof}
$(1)$ See at the end of the previous subsection.

$(2)$ Mappings~(\ref{Rij}) written in terms of the functions $f_i$, $g_i$ get exactly the desired form.

$(3)$ See Proposition~\ref{prop1}.

$(4)$ See Proposition~\ref{prop1}.

$(5)$ Because mappings $R_{ij}^{{\bf p}_{ij}}$, for generic parameter sets ${\bf p}_{ij}$, belong to the $[2:2]$ subclass, we expect at most $8$ singular points, $4$ singular points from the first fraction of the map and $4$ from the second. By direct calculation we show that the singular points of the first and the second fraction of $R_{ij}^{{\bf p}_{ij}}$ coincide. Moreover, $P^m_{ij}$, $m=1,\ldots,4$ are the singular points of the maps $R_{ij}^{{\bf p}_{ij}}$, i.e.,
\begin{gather*}
R_{ij}^{{\bf p}_{ij}}\colon \ P^m_{ij}\mapsto \left(x_1,\ldots,x_{i-1},\frac{0}{0},x_{i+1},\ldots,x_{j-1},\frac{0}{0},x_{j+1},\ldots, x_k\right).
\end{gather*}
Note that the values of the invariants $H_i$ at the singular points $P^m_{ij}$ are undetermined, i.e., $H_1\big(P^m_{ij}\big)=\frac{0}{0}$, $m=1,2$,
$H_2\big(P^m_{ij}\big)=\frac{0}{0}$, $m=3,4$. For the fixed points $Q^m_{ij}$, $m=1,\ldots,4$ it holds
$R_{ij}^{{\bf p}_{ij}}\colon Q^m_{ij}\mapsto Q^m_{ij}$. Note also that $H_1\big(Q^1_{ij}\big)=0$, $H_1\big(Q^2_{ij}\big)=\infty$, $H_2\big(Q^3_{ij}\big)=0$, $H_2\big(Q^4_{ij}\big)=\infty$.

$(6)$ Introducing the new variables $y_i$, $y_j$, $i\neq j=1,\ldots, k$ though
\begin{gather*}
\operatorname{CR}[x_i,a_i/b_i,c_i/d_i,A_i/B_i]=\operatorname{CR}[y_i,0,1,\infty],\\
 \operatorname{CR}[x_j,c_j/d_j,a_j/b_j,C_j/D_j]=\operatorname{CR}[y_j,\infty,1,0],
\end{gather*}
after a re-parametrization mappings $R_{ij}$ gets exactly the form of the $H_{\rm I}$ map. Here, with $\operatorname{CR}[a,b,c,d]$ we denote the cross-ratio of $4$ points, namely
\begin{gather*}
\operatorname{CR}[a,b,c,d]:=\frac{(a-c)(b-d)}{(a-d)(b-c)}.\tag*{\qed}
\end{gather*}\renewcommand{\qed}{}
\end{proof}

Each one of the maps $R_{ij}$ has a set of singularities which consists of $4$ distinct points. With appropriate limits we are allowed to merge some of the singularities and obtain Yang--Baxter maps which are not $\text{(M{\"o}b)}^2$ equivalent with the original one.

By setting $C_i=\epsilon A_i$, $D_i=\epsilon B_i$, $A_j=\epsilon C_j$, $B_j=\epsilon D_j$ and letting $\epsilon \rightarrow 0$ the singular points~$P^4_{ij}$ and~$P^3_{ij}$ merge. The resulting maps, under a re-parametrization, coincide with the ones obtained in the multiplicative/additive case (see Section~\ref{subsection3.2}), hence are $\text{(M\"ob)}^2$ equivalent with the $H_{\rm II}$ Yang--Baxter map. The same result can be obtained by merging~$P^2_{ij}$ and~$P^1_{ij}$. Note that mer\-ging~$P^4_{ij}$ with $P^2_{ij}$ or $P^4_{ij}$ with $P^1_{ij}$ is not of interest since the resulting maps are trivial.

By further setting $c_i=\epsilon a_i$, $d_i=\epsilon b_i$, $a_j=\epsilon c_j$, $b_j=\epsilon d_j$ and letting $\epsilon \rightarrow 0$ the singular points~$P^2_{ij}$ and $P^1_{ij}$ merge as well. The resulting maps, under a re-parametrization, coincide with the ones obtained in the additive/additive case (see Section~\ref{subsection3.3}), hence are $\text{(M{\"o}b)}^2$ equivalent with the $H_{\rm III}^A$ Yang--Baxter map.
Any further merging of singularities leads to trivial maps.

\begin{Remark}\label{remark:3.3}
An interesting observation is that if we impose that the fixed points $Q^4_{ij}$ of the maps $R_{ij}$ coincide with the singular points $P^2_{ij}$ or the fixed points $Q^4_{ij}$ coincide with $P^1_{ij}$, we obtain maps which belong to the $[1:1]$ subclass of maps. The same is true if we demand that the fixed points $Q^1_{ij}$ coincide with the singular points $P^3_{ij}$ or if the fixed points $Q^1_{ij}$ coincide with the singular points $P^4_{ij}$,
\end{Remark}

\begin{Remark}For generic sets of parameters ${\bf p}_{ij}$, each one of the $\binom{k}{2}$ maps $R_{ij}^{{\bf p}_{ij}}$, is $\text{(M\"ob)}^2$
equivalent to the $H_{\rm I}$ Yang--Baxter map. For degenerate choices of the sets ${\bf p}_{ij}$, this is no longer the case. Hence, in that respect, mappings
$R_{ij}^{{\bf p}_{ij}}$ are more general than the $H_{\rm I}$ map since they include degenerate cases as well. In the same respect $Q_{\rm V}$ \cite{Viallet:2009}, the rational version of the discrete Krichever--Novikov equation $Q_4$ \cite{Adler:1998}, is more general.
\end{Remark}

\begin{Example}[$k=3$] \label{ex111} For $k=3$, the invariants $H_1=f_1f_2f_3$, $H_2=g_1g_2g_3$ are functions of $3$ variables with $24$ parameters, $3$ of them are essential. Without loss of generality, after removing the redundancy of the parameters, the invariants $H_1$, $H_2$ can be cast into the form
\begin{gather*}
H_1=x_1x_2x_3,\qquad H_2=\frac{x_1-p_1}{x_1-1}\frac{x_2-p_2}{x_2-1}\frac{x_3-p_3}{x_3-1}.
\end{gather*}
Then each of the mappings $R_{ij}$, $i\neq j \in \{1,2,3\}$ is exactly the $H_{\rm I}$ Yang--Baxter map. The $H_{\rm I}$ Yang--Baxter map explicitly reads
$H_{\rm I}\colon (u,v)\mapsto (U,V)$ where
\begin{gather} \label{H1}
 U=v Q, \qquad V=uQ^{-1},
 \qquad Q=\frac{(\alpha-1)uv+(\beta-\alpha)u+\alpha(1-\beta)}{(\beta-1)uv+(\alpha-\beta)v+\beta(1-\alpha)}.
\end{gather}
By the identifications $u\equiv x_i$, $u\equiv x_j$, $\alpha\equiv p_i$ and $\beta\equiv p_j$, from~(\ref{H1}) we recover the maps~$R_{ij}$.

The maps $\phi_i\colon (x_1,x_2,x_3)\mapsto (X_1,X_2,X_3)$ where $X_l=x_l$ $\forall\, l\neq i$ and $X_i=\frac{p_i}{x_i}$, $i=1,2,3$
and the maps $\psi_i\colon (x_1,x_2,x_3)\mapsto (X_1,X_2,X_3)$ where $X_l=x_l$ $\forall\, l\neq i$ and $X_i=\frac{x_i-p_i}{x_i-1}$, $i=1,2,3$ satisfy
\begin{gather*}
H_1 \phi_1 \phi_2 \phi_3 =\frac{p_1 p_2 p_3}{H_1}, \qquad\! H_2 \phi_1 \phi_2 \phi_3 =\frac{p_1 p_2 p_3}{H_2}, \qquad\! H_1 \psi_1 \psi_2 \psi_3= H_2, \qquad\! H_2 \psi_1 \psi_2 \psi_3=H_1.
\end{gather*}
The maps $\phi_i$ and $\psi_i$ have a special role in \cite{pap3-2010} since though them the $H_{\rm I}$ map was derived out of the $F_{\rm I}$ Yang--Baxter map. We will discuss more about these maps in the next Section. We just quickly recall that $\phi_1 R_{12} \phi_2$ is exactly the $F_{\rm I}$ Yang--Baxter map.

\begin{Remark}We have to remark that with loss of generality, mappings $R_{ij}$ can belong on a~different subclasses than the $[2:2]$ subclass of maps that the $H_{\rm I}$ map belongs to. For example, for
\begin{gather*}
H_1=(x_1-p_1)(x_2-p_2)(x_3-p_3), \qquad H_2=\frac{x_1-p_1}{x_1}\frac{x_2}{x_2-p_2}\frac{x_3}{x_3-1},
\end{gather*}
$R_{12}$ is the Hirota's KdV map (see \cite{PKMN2}) that belongs on the subclass $[1:1]$ and $R_{13}$, $R_{23}$ are maps which belong to the subclass $[2:1]$. Explicitly the maps read
\begin{gather*}
R_{12}\colon \ (x_1,x_2,x_3)\mapsto {\displaystyle \left(\frac{p_1 (p_2 x_1 + p_1 x_2 - x_1 x_2)}{p_2 x_1}, \frac{p_2 (p_2 x_1 + p_1 x_2 - x_1 x_2)}{ p_1 x_2}, x_3\right)}, \\
R_{13}\equiv S_{13}\colon \ (x_1,x_2,x_3)\mapsto {\displaystyle \left(\frac{p_1 (-1 + x_3) (p_3 x_1 + p_1 x_3 - x_1 x_3)}{-p_3 x_1 - p_1 x_3 + p_1 p_3 x_3 +
 x_1 x_3}, x_2, \frac{p_3 x_1 + p_1 x_3 - x_1 x_3}{p_1 x_3}\right)}, \\
R_{23}\equiv T_{23}\colon \ (x_1,x_2,x_3)\mapsto {\displaystyle \left(x_1, \frac{p_2 x_3 (-p_2 + p_3 x_2 + p_2 x_3 - x_2 x_3)}{-p_2 p_3 + p_3 x_2 + p_2 p_3 x_3 -
 x_2 x_3}, \frac{x_2 (-p_3 + x_3)}{p_2 (-1 + x_3)}\right)}.
\end{gather*}
The Hirota's KdV map entwines with $S_{13}$ and $T_{23}$, since $R_{12}S_{13}T_{23}=T_{23}S_{13}R_{12}$ holds.
\end{Remark}
\end{Example}

\begin{Example}[$k\geq4$]
For $k=4$ the invariants depend on $32$ parameters and only $6$ of them are essential. Without loss of generality they can be cast into the form
\begin{gather*}
H_1=x_1x_2x_3x_4,\qquad H_2=\frac{x_1-p_1}{x_1-1}\frac{x_2-p_2}{x_2-1}\frac{x_3-p_3}{x_3-1}\frac{\alpha_4-\beta_4 x_4}{\beta_4-\gamma_4x_4}.
\end{gather*}
For $k>4$ the invariants depend on $8k$ parameters and only $3k-6$ of them are essential. Without loss of generality they can be cast into the form
\begin{gather*}
H_1=\prod_{i=1}^k x_i,\qquad H_2=\frac{x_1-p_1}{x_1-1}\frac{x_2-p_2}{x_2-1}\frac{x_3-p_3}{x_3-1}\prod_{i=4}^k \frac{\alpha_i-\beta_i x_i}{\beta_i-\gamma_i x_i}.
\end{gather*}
\end{Example}

\subsection{Multiplicative/additive separability of variables} \label{subsection3.2}

\begin{Proposition} \label{prop3}
Consider the multiplicative/additive separability of variables of the inva\-riants~$H_1$ and $H_2$ $($see \eqref{multi/add}$)$. Consider also the following sets of parameters
\begin{gather*}
{\bf p}_{ij}:={\bf p}_{i}\cup {\bf p}_{j}, \qquad \mbox{where} \qquad {\bf p}_{i}:=\left\{a_i,b_i,c_i,d_i,A_i,B_i,C_i,D_i \right\},\qquad i<j \in \{1,2,\ldots, k\}
\end{gather*}
and the functions
\begin{gather*}
f_i:=\frac{a_i-b_ix_i}{c_i-d_ix_i}, \qquad g_i:=\frac{A_i-B_ix_i}{C_i-D_ix_i}, \qquad i=1,\ldots, k.
\end{gather*}
The following holds:
\begin{enumerate}\itemsep=0pt
\item[$1.$] The invariants $H_1=\prod\limits_{i=1}^k f_i$, $H_2=\sum\limits_{i=1}^k g_i$ depend on $8k$ parameters. Only $3k-6$ of them are essential.
\item[$2.$] Mappings $R_{ij}$ explicitly read
\begin{gather*}
R_{ij}\colon \ (x_1,x_2,\ldots, x_k)\mapsto (X_1,X_2,\ldots, X_k),
\end{gather*}
where $X_l=x_l$ $\forall\, l\neq i,j$ and $X_i$, $X_j$ are given by the formulae
\begin{gather*}
X_i=x_i- 2\frac{\left|\begin{matrix}
 f_if_j'& f_i'f_j\\
 g_j'&g_i'\end{matrix}\right| }
 { \left|\begin{matrix}
 f_i'f_j& g_i'\\
 \dfrac{f_j'}{f_i'}f_if_i''+\dfrac{f_i'}{f_j'}f_jf_j''-2f_i'f_j'&\dfrac{g_j'}{g_i'}g_i''+\dfrac{g_i'}{g_j'}g_j''\end{matrix}\right| }, \\
X_j=x_j+ 2\frac{\left|\begin{matrix}
 f_if_j'& f_i'f_j\\
 g_j'&g_i'\end{matrix}\right| }
 { \left|\begin{matrix}
 f_j'f_i& g_j'\\
 \dfrac{f_j'}{f_i'}f_if_i''+\dfrac{f_i'}{f_j'}f_jf_j''-2f_i'f_j'&\dfrac{g_j'}{g_i'}g_i''+\dfrac{g_i'}{g_j'}g_j''\end{matrix}\right| },
\end{gather*}
where $f_l'\equiv \frac{\partial f_l}{\partial x_l}$, $g_l'\equiv \frac{\partial g_l}{\partial x_l}$, $g_l''\equiv \frac{\partial^2 g_l}{\partial x_l^2}$, etc. Note that in the expressions of $X_i$, $X_j$ appears only the coordinates $x_i$, $x_j$ and the parameters ${\bf p}_{ij}$. From further on we denote the maps~$R_{ij}$ as~$R_{ij}^{{\bf p}_{ij}}$, in order to stress this separability feature.
\item[$3.$] Mappings $R_{ij}^{{\bf p}_{ij}}$ are anti-measure preserving with densities $m_1=n^1 d^2$, $m_2=n^2 d^1$, where $n^i$,~$d^i$ the numerators and the denominators respectively, of the invariants $H_i$, $i=1,2$.
\item[$4.$] Mappings $R_{ij}^{{\bf p}_{ij}}$ satisfy the Yang--Baxter identity
\begin{gather*}
R_{ij}^{{\bf p}_{ij}} R_{ik}^{{\bf p}_{ik}} R_{jk}^{{\bf p}_{jk}}=R_{jk}^{{\bf p}_{jk}} R_{ij}^{{\bf p}_{ij}} R_{ij}^{{\bf p}_{ij}}.
\end{gather*}
\item[$5.$] Mappings $R_{ij}^{{\bf p}_{ij}}$ are involutions with the sets of singularities
\begin{gather*}
\Sigma_{ij}=\big\{P^1_{ij},P^2_{ij},P^3_{ij}\big\}=
\left\{\left(\frac{a_i}{b_i},\frac{c_j}{d_j}\right),\left(\frac{c_i}{d_i},\frac{a_j}{b_j}\right),\left(\frac{C_i}{D_i},\frac{C_j}{D_j}\right)^2\right\},
\end{gather*}
where the superscript $2$ in $P^3_{ij}$ denotes that these singular points appears with multiplicity $2$.
In the formulae for $P_{ij}^m$, $m=1,\ldots, 3$, we have suppressed the dependency on the remaining variables. For example, with
$P_{ij}^1=\big(\frac{a_i}{b_i},\frac{c_j}{d_j}\big)$ we denote $\big(x_1,\ldots,x_{i-1},\frac{a_i}{b_i},x_{i+1},\ldots,x_{j-1},\frac{c_j}{d_j},\allowbreak x_{j+1},\ldots, x_k\big)$ and similarly for the remaining $P_{ij}^m$.
\item[$6.$] Each one of the maps $R_{ij}^{{\bf p}_{ij}}$ is $\text{\rm (M\"ob)}^2$ equivalent to the $H_{\rm II}$ Yang--Baxter map.
\end{enumerate}
\end{Proposition}
\begin{proof}
The proof follows similarly to the proof of Proposition \ref{prop2}.
\end{proof}
\begin{Example}[$k\geq3$]
For $k=3$, the invariants $H_1=f_1f_2f_3$, $H_2=g_1+g_2+g_3$ are functions of $3$ variables with $24$ parameters, $3$ of them are essential. Without loss of generality, after removing the redundancy of the parameters, the invariants $H_1,H_2$ can be cast into the form
\begin{gather*}
H_1=\frac{x_1-p_1}{x_1}\frac{x_2-p_2}{x_2}\frac{x_3-p_3}{x_3},\qquad H_2=x_1+x_2+x_3.
\end{gather*}
Then each of the mappings $R_{ij}$, $i\neq j \in \{1,2,3\}$ is exactly the $H_{\rm II}$ Yang--Baxter map.

For $k>3$ the invariants depend on $8k$ parameters and only $3k-6$ of them are essential. Without loss of generality they can be cast into the form
\begin{gather*}
H_1=\frac{x_1-p_1}{x_1}\frac{x_2-p_2}{x_2}\frac{x_3-p_3}{x_3}\prod_{i=4}^k \frac{\alpha_i-\beta_i x_i}{\beta_i-\gamma_i x_i},\qquad H_2=\sum_{i=1}^k x_i.
\end{gather*}
\end{Example}

\subsection{Additive/additive separability of variables} \label{subsection3.3}

\begin{Proposition} \label{prop4}
Consider the additive/additive separability of variables of the invariants $H_1$ and $H_2$ $($see \eqref{add/add}$)$. Consider also the following sets of parameters
\begin{gather*}
{\bf p}_{ij}:={\bf p}_{i}\cup {\bf p}_{j}, \qquad\!\! \mbox{where} \qquad\!\!\! {\bf p}_{i}:=\left\{a_i,b_i,c_i,d_i,A_i,B_i,C_i,D_i \right\}, \qquad\!\! i\neq j<j \in \{1,2,\ldots, k\}
\end{gather*}
and the functions
\begin{gather*}
f_i:=\frac{a_i-b_ix_i}{c_i-d_ix_i}, \qquad g_i:=\frac{A_i-B_ix_i}{C_i-D_ix_i}, \qquad i=1,\ldots, k.
\end{gather*}
The following holds:
\begin{enumerate}\itemsep=0pt
\item[$1.$] The invariants $H_1=\prod\limits_{i=1}^k f_i$, $H_2=\sum\limits_{i=1}^k g_i$ depend on $8k$ parameters. Only $3k-6$ of them are essential.
\item[$2.$] Mappings $R_{ij}$ explicitly read
\begin{gather*}
R_{ij}\colon \ (x_1,x_2,\ldots, x_k)\mapsto (X_1,X_2,\ldots, X_k),
\end{gather*}
where $X_l=x_l$ $\forall\, l\neq i,j$ and $X_i$, $X_j$ are given by the formulae
\begin{gather*}
X_i=x_i- 2\frac{\left|\begin{matrix}
 f_j'& f_i'\\
 g_j'&g_i'\end{matrix}\right| }
 { \left|\begin{matrix}
 f_i'& g_i'\\
 \dfrac{f_j'}{f_i'}f_i''+\dfrac{f_i'}{f_j'}f_j''&\dfrac{g_j'}{g_i'}g_i''+\dfrac{g_i'}{g_j'}g_j''\end{matrix}\right|}, \\
X_j=x_j+ 2\frac{\left|\begin{matrix}
 f_j'& f_i'\\
 g_j'&g_i'\end{matrix}\right| }
 { \left|\begin{matrix}
 f_j'& g_j'\\
 \dfrac{f_j'}{f_i'}f_i''+\dfrac{f_i'}{f_j'}f_j''&\dfrac{g_j'}{g_i'}g_i''+\dfrac{g_i'}{g_j'}g_j''\end{matrix}\right|},
\end{gather*}
where $f_l'\equiv \frac{\partial f_l}{\partial x_l}$, $g_l'\equiv \frac{\partial g_l}{\partial x_l}$, $g_l''\equiv \frac{\partial^2 g_l}{\partial x_l^2}$, etc. Note that in the expressions of $X_i$, $X_j$ appears only the coordinates $x_i$, $x_j$ and the parameters ${\bf p}_{ij}$. From further on we denote the maps~$R_{ij}$ as~$R_{ij}^{{\bf p}_{ij}}$, in order to stress this separability feature.
\item[$3.$] Mappings $R_{ij}^{{\bf p}_{ij}}$ are anti-measure preserving with densities $m_1=n^1 d^2$, $m_2=n^2 d^1$, where $n^i$,~$d^i$ the numerators and the denominators respectively, of the invariants $H_i$, $i=1,2$.
\item[$4.$] Mappings $R_{ij}^{{\bf p}_{ij}}$ satisfy the Yang--Baxter identity
\begin{gather*}
R_{ij}^{{\bf p}_{ij}} R_{ik}^{{\bf p}_{ik}} R_{jk}^{{\bf p}_{jk}}=R_{jk}^{{\bf p}_{jk}} R_{ij}^{{\bf p}_{ij}} R_{ij}^{{\bf p}_{ij}}.
\end{gather*}
\item[$5.$] Mappings $R_{ij}^{{\bf p}_{ij}}$ are involutions with the sets of singularities
\begin{gather*}
\Sigma_{ij}=\big\{P^1_{ij},P^2_{ij}\big\}=
\left\{\left(\frac{c_i}{d_i},\frac{c_j}{d_j}\right)^2,\left(\frac{C_i}{D_i},\frac{C_j}{D_j}\right)^2\right\},
\end{gather*}
where the superscript $2$ in $P^1_{ij}$ and $P^2_{ij}$ denotes that these singular points appears with multiplicity~$2$.
In the formulae for $P_{ij}^m$, $m=1,\ldots, 2$, we have suppressed the dependency on the remaining variables. For example, with
$P_{ij}^1=\big(\frac{c_i}{d_i},\frac{c_j}{d_j}\big)$ we denote $\big(x_1,\ldots,x_{i-1},\frac{c_i}{d_i},x_{i+1},\ldots,\allowbreak x_{j-1},\frac{c_j}{d_j},x_{j+1},\ldots, x_k\big)$ and similarly for~$P_{ij}^2$.
\item[$6.$] Each one of the maps $R_{ij}^{{\bf p}_{ij}}$ is $\text{\rm (M\"ob)}^2$ equivalent to the $H_{\rm III}^A$ Yang--Baxter map.
\end{enumerate}
\end{Proposition}
\begin{proof}The proof follows similarly to the proof of Proposition \ref{prop2}.
\end{proof}
\begin{Example}[$k\geq3$]
For $k=3$, the invariants $H_1=f_1+f_2+f_3$, $H_2=g_1+g_2+g_3$ are functions of $3$ variables with $24$ parameters, $3$ of them are essential. Without loss of generality, after removing the redundancy of the parameters, the invariants $H_1$, $H_2$ can be cast into the form:
\begin{gather*}
H_1=\frac{1}{x_1}+\frac{1}{x_2}+\frac{1}{x_3},\qquad H_2=p_1x_1+p_2x_2+p_3x_3.
\end{gather*}
Then each of the mappings $R_{ij}$, $i\neq j \in \{1,2,3\}$ is exactly the $H_{\rm III}^A$ Yang--Baxter map.

For $k>3$ the invariants depend on $8k$ parameters and only $3k-6$ of them are essential. Without loss of generality they can be cast into the form
\begin{gather*}
H_1=\sum_{i=1}^k \frac{1}{x_i},\qquad H_2=p_1x_1+p_2x_2+p_3x_3+\sum_{i=4}^k \frac{\alpha_i-\beta_i x_i}{\beta_i-\gamma_i x_i}.
\end{gather*}
\end{Example}

\section{Entwining Yang--Baxter maps} \label{Section4}
Following \cite{Kouloukas:2011}, three different maps $ S$, $T$, $U$ are called {\it entwining Yang--Baxter maps} if they satisfy
\begin{gather*}
S_{12}T_{13}U_{23}=U_{23}T_{13}S_{12}.
\end{gather*}
 We consider two maps to be different if they are not $\text{(M\"ob)}^2$ equivalent. Hence, in order to ensure that we have different maps we require that at least one of the maps $ S$, $T$, $U$ either belongs to a different subclass than the remaining ones or it has different singularity pattern (even if it belongs to the same subclass with the remaining ones) or it has different periodicity. In what follows we present two methods to obtain entwining maps. The first one is based on degeneracy, i.e., we construct maps which belong to different subclasses and we obtain entwining maps associated with the $H_{\rm I}$, $H_{\rm II}$ and $H_{\rm III}^A$ families of maps. The second one is based on the symmetries of the $H$-list of Yang--Baxter maps and we obtain entwining maps for all members of the $H$-list.

\subsection{Degeneracy and entwining Yang--Baxter maps}
In Section~\ref{subsection3.1} it was shown that for $k=3$ and for the multiplicative/multiplicative case, the invariants $H_1$, $H_2$ depend on $3$ essential parameters. Without loss of generality they read
\begin{gather*} 
H_1=x_1x_2x_3,\qquad H_2=\frac{x_1-p_1}{x_1-1}\frac{x_2-p_2}{x_2-1}\frac{x_3-p_3}{x_3-1}.
\end{gather*}
The associated maps $R_{12}$, $R_{13}$ and $R_{23}$ which preserve the invariants have exactly the form of the $H_{\rm I}$ map. In order to obtain entwining maps associated with the $H_{\rm I}$ map, we consider
\begin{gather*}
H_1=x_1x_2x_3,\qquad H_2=\frac{x_1-p_1}{x_1-1}\frac{x_2-p_2}{x_2-1}\frac{\alpha_3-\beta_3x_3}{\beta_3-\gamma_3x_3}.
\end{gather*}
For these invariants, $R_{12}$ is exactly the $H_{\rm I}$ map and for generic $\alpha_3$, $\beta_3$, $\gamma_3$ mappings~$R_{13}$ and~$R_{23}$ are~$\text{(M\"ob)}^2$ equivalent to the $H_{\rm I}$.
In order to obtain entwining maps we need to violate this~$\text{(M\"ob)}^2$ equivalency of the maps $R_{13}$ and $R_{23}$ with the $H_{\rm I}$ map. This is achieved by violating the generality, e.g., setting $\alpha_3=0$ or $\beta_3=0$, the maps $R_{13}$ and $R_{23}$, belongs to different subclasses than the $H_{\rm I}$ map does. Working similarly for the $H_{\rm II}$ map we find $1$ family of maps which entwine with the latter without being $\text{(M\"ob)}^2$ equivalent. Finally, for $H_{\rm III}^A$ we find also $1$ family of entwining maps which are not $\text{(M\"ob)}^2$ equivalent with the latter. Our results are presented in Propositions~\ref{eHI}--\ref{eHIII}.

\begin{table}[!hb]\centering
\caption{Entwining maps associated with the $H_{\rm I}$ Yang--Baxter map through degeneracy.} \label{table1}\vspace{1mm}
\begin{tabular}{l|c|l}
\hline
map & $(u,v)\mapsto (U,V)$ & subclass \\ \hline
${\rm e}^aH_{\rm I}$
 & $U={\displaystyle \frac{\alpha (1-u)+\beta(\alpha-1)uv}{\alpha-u}}$, $V={\displaystyle \frac{uv(\alpha-u)}{\alpha (1-u)+\beta(\alpha-1)uv}}$\tsep{7pt}\bsep{7pt}
 & $[1:2]$
 \\
${\rm e}^bH_{\rm I}$
 & $U={\displaystyle \frac{u-\alpha}{u-1}}$, $V={\displaystyle \frac{uv(u-1)}{u-\alpha} }$\bsep{7pt}
 & $[0:2]$
 \\
\hline
\end{tabular}
\end{table}

\begin{Proposition} \label{eHI}The $H_{\rm I}$ Yang--Baxter map entwines with the maps $ {\rm e}^aH_{\rm I}$ and ${\rm e}^bH_{\rm I}$ of Table~{\rm \ref{table1}} according to the entwining relation
\begin{gather*}
S_{12}T_{13}T_{23}=T_{23}T_{13}S_{12},
\end{gather*}
where $S_{12}$ is the $H_{\rm I}$ map acting on the $(1,2)$-coordinates, $T_{13}$ and $T_{23}$ are ${\rm e}^aH_{\rm I}$ acting on $(1,3)$ and $(2,3)$ coordinates respectively, or ${\rm e}^bH_{\rm I}$ acting on $(1,3)$ and $(2,3)$ coordinates respectively.
\end{Proposition}

\begin{proof}Starting with the invariants
\begin{gather*}
H_1=x_1x_2x_3,\qquad H_2=\frac{x_1-p_1}{x_1-1}\frac{x_2-p_2}{x_2-1}\frac{a-b x_3}{b-c x_3},
\end{gather*}
the map $R_{12}$ is exactly the $H_{\rm I}$ map. By setting $a=0$, $R_{13}$ and $R_{23}$ takes the form of ${\rm e}^aH_{\rm I}$ of Table~\ref{table1} (where $\beta\equiv c/b$). The map ${\rm e}^aH_{\rm I}$ is of subclass $[1:2]$ so clearly non-$\text{(M\"ob)}^2$ equivalent to~$H_{\rm I}$.
By setting $b=0$, $R_{13}$ and $R_{23}$ takes the form of ${\rm e}^bH_{\rm I}$ of Table~\ref{table1} (where $\beta\equiv a/c$). The map ${\rm e}^bH_{\rm I}$ is of subclass $[0:1]$ so clearly non-$\text{(M\"ob)}^2$ equivalent to $H_{\rm I}$ or to ${\rm e}^aH_{\rm I}$.
Finally, by setting $c=0$, mappings $R_{13}$ and $R_{23}$ are $\text{(M\"ob)}^2$ equivalent to ${\rm e}^aH_{\rm I}$.
\end{proof}

\begin{Proposition} \label{eHII}
The $H_{\rm II}$ Yang--Baxter map entwines with the map of Table~{\rm \ref{table2}}
 according to the entwining relation
\begin{gather*}
S_{12}T_{13}T_{23}=T_{23}T_{13}S_{12},
\end{gather*}
where $S_{12}$ is the $H_{\rm II}$ map acting on the $(1,2)$-coordinates, $T_{13}$ and $T_{23}$ are ${\rm e}^bH_{\rm II}$ acting on $(1,3)$ and $(2,3)$ coordinates respectively.
\end{Proposition}

\begin{table}[!h]\centering
\caption{Entwining maps associated with the $H_{\rm II}$ Yang--Baxter map though degeneracy.} \label{table2}\vspace{1mm}
\begin{tabular}{c|c|c}
\hline
map & $(u,v)\mapsto (U,V)$ & subclass \\ \hline
${\rm e}^bH_{\rm II}$
 & $U={\dis \frac{\alpha v}{\alpha-u} }$, $V={\dis u\frac{\alpha-u-v}{\alpha-u} } $\tsep{7pt}\bsep{7pt}
 & $[1:1]$
 \\
\hline
\end{tabular}
\end{table}

\begin{proof}Starting with the invariants
\begin{gather*}
H_1=x_1+x_2+x_3,\qquad H_2=\frac{x_1-p_1}{x_1}\frac{x_2-p_2}{x_2}\frac{a-b x_3}{b-c x_3},
\end{gather*}
the map $R_{12}$ is exactly the $H_{\rm II}$ map. By setting $a=0$, $R_{13}$ and $R_{23}$ are $\text{(M\"ob)}^2$ equivalent to the $H_{\rm II}$ map.
By setting $b=0$, $R_{13}$ and $R_{23}$ takes the form of ${\rm e}^bH_{\rm II}$ of Table~\ref{table2}. The map ${\rm e}^bH_{\rm II}$ is of subclass $[1:1]$ so clearly non-$\text{(M\"ob)}^2$ equivalent to the $H_{\rm II}$ map.
Finally, by setting $c=0$, mappings $R_{13}$ and $R_{23}$ are $\text{(M\"ob)}^2$ equivalent to~${\rm e}^bH_{\rm II}$.
\end{proof}

\begin{Proposition} \label{eHIII}
The $H_{\rm III}^A$ Yang--Baxter map entwines with the map of Table~{\rm \ref{table3}}
 according to the entwining relation
\begin{gather*}
S_{12}T_{13}T_{23}=T_{23}T_{13}S_{12},
\end{gather*}
where $S_{12}$ is the $H_{\rm III}^A$ map acting on the $(1,2)$-coordinates, $T_{13}$ and $T_{23}$ are ${\rm e}^bH_{\rm III}^A$ acting on $(1,3)$ and $(2,3)$ coordinates respectively.
\end{Proposition}

\begin{table}[!h]\centering
\caption{Entwining maps associated with the $H_{\rm III}^A$ Yang--Baxter map though degeneracy.} \label{table3}\vspace{1mm}
\begin{tabular}{c|c|c}
\hline
map & $(u,v)\mapsto (U,V)$ & subclass \\ \hline
${\rm e}^bH_{\rm III}^A$
 & $U={\dis \frac{\beta}{\alpha} u }$, $V={\dis \frac{\beta u v}{\beta (u+v)-\alpha u^2v} } $\tsep{7pt}\bsep{7pt}
 & $[0:2]$
 \\
\hline
\end{tabular}
\end{table}

\begin{proof}Starting with the invariants
\begin{gather*}
H_1=x_1+x_2+x_3,\qquad H_2=p_1x_1+p_2x_2+\frac{a-b x_3}{b-c x_3},
\end{gather*}
the map $R_{12}$ is exactly the $H_{\rm III}^A$ map. By setting $a=0$, $R_{13}$ and $R_{23}$ are $\text{(M\"ob)}^2$ equivalent to the $H_{\rm III}^A$ map.
By setting $b=0$ and $R_{13}$ and $R_{23}$ takes the form of ${\rm e}^bH_{\rm III}^A$ of Table~\ref{table3} (where $\beta=a/c$). The map ${\rm e}^bH_{\rm III}^A$ is of subclass $[0:2]$ so clearly non-$\text{(M\"ob)}^2$ equivalent to the $H_{\rm III}^A$ map.
Finally, by setting $c=0$, mappings $R_{13}$ and $R_{23}$ are $\text{(M\"ob)}^2$ equivalent to the $H_{\rm III}^A$ map.
\end{proof}

In the following subsection we are using the notion of {\it symmetry of Yang--Baxter maps} in order to generate entwining maps

\subsection{Symmetries of Yang--Baxter maps and the entwining property}
The notion of symmetry in the context of Yang--Baxter maps was introduced in \cite{pap3-2010}.
 \begin{Definition} \label{def3}
An involution $\phi\colon \mathbb{CP}^1\mapsto \mathbb{CP}^1$ is a symmetry of the Yang--Baxter map $R\colon \mathbb{CP}^1\times \mathbb{CP}^1\mapsto \mathbb{CP}^1\times \mathbb{CP}^1$ if it holds
\begin{gather*}
\phi_1\phi_2R_{12}=R_{12}\phi_1\phi_2,
\end{gather*}
where $\phi_1$ is the involution that acts as $\phi$ to the first factor of the cartesian product $\mathbb{CP}^1\times \mathbb{CP}^1$ and $\phi_2$ is the involution that acts as $\phi$ to the second factor of the cartesian product.
\end{Definition}

 Let $m<n\in \{1,\ldots, k\}$, $k\geq 3$ fixed. A direct consequence of the previous definition is that if $\phi$ is a symmetry of the Yang--Baxter map $R$, then the map $\phi_m R_{mn} \phi_n$ is a new Yang--Baxter map since it is not $\text{(M\"ob)}^2$ equivalent with $R_{mn}$. By finding the symmetries of the $F$-list of Yang--Baxter maps, the authors of~\cite{pap3-2010} derived the $H$-list of Yang--Baxter maps. Clearly the symmetries of the $F$-list are symmetries of the $H$-list and vice versa.

\begin{Theorem} \label{thrm1}
Let $\phi$ a symmetry of a Yang--Baxter map $R$ and let $\phi_0$ the identity map, i.e., $\phi_0\colon (x_1,\ldots, x_k)\mapsto (x_1,\ldots, x_k)$. Out of the possible $4^3$ entwining relations of the form
\begin{gather} \label{ent-test}
R_{12} \phi_i R_{13} \phi_j R_{23} \phi_k=R_{23} \phi_k R_{13} \phi_j R_{12} \phi_i,\qquad i,j,k\in\{0,1,2,3\},
\end{gather}
apart the Yang--Baxter relation that holds, only the following three entwining relations holds
\begin{gather} \label{ent-re1}
R_{12} R_{13} \phi_1 R_{23} \phi_2= R_{23} \phi_2 R_{13}\phi_1 R_{12}, \\ \label{ent-re2}
R_{12} \phi_2 R_{13} \phi_3 R_{23} = R_{23} R_{13}\phi_3 R_{12}\phi_2, \\ \label{ent-re3}
R_{12} \phi_2 R_{13} \phi_2 R_{23} \phi_2= R_{23} \phi_2 R_{13}\phi_2 R_{12}\phi_2.
\end{gather}
\end{Theorem}
\begin{proof}To show that only the entwining relations (\ref{ent-re1}), (\ref{ent-re2}), (\ref{ent-re3}) holds,
we start with
\begin{gather*}
R_{12} \phi_i R_{13} \phi_j R_{23} \phi_k=R_{23} \phi_k R_{13} \phi_j R_{12} \phi_i, \qquad i,j,k\in\{0,1,2,3\}.
\end{gather*}
 By direct calculations, we prove that if the Yang--Baxter relation holds out of the $4^3$ different relations (\ref{ent-test}), only (\ref{ent-re1}), (\ref{ent-re2}), (\ref{ent-re3}) holds.

For example let us show that~(\ref{ent-re1}) holds. We have
\begin{gather} \label{prf1}
R_{12} R_{13} \phi_1 R_{23} \phi_2=R_{12} R_{13} R_{23} \phi_1 \phi_2= R_{23}R_{13}R_{12}\phi_1 \phi_2,
\end{gather}
since $\phi_1$ commutes with $R_{23}$ and the Yang--Baxter relation $R_{12} R_{13} R_{23} =R_{23} R_{13} R_{12}$ holds. But due to the symmetry we have
$ R_{12}\phi_1 \phi_2=\phi_1 \phi_2 R_{12}$ so~(\ref{prf1}) reads
\begin{gather*}
R_{23}R_{13}R_{12}\phi_1 \phi_2=R_{23}R_{13}\phi_1 \phi_2 R_{12}=R_{23}\phi_2 R_{13} \phi_1 R_{12}
\end{gather*}
and that completes the proof that~(\ref{ent-re1}) holds. For the remaining relations we work similarly for their proof.
\end{proof}

Note that any of the entwining relations (\ref{ent-re1}), (\ref{ent-re2}) and~(\ref{ent-re3}), is uniquely described by the symmetries $\phi_i$, $\phi_j$, $\phi_k$ that take part in this relation. For example in (\ref{ent-re1}) the symmetries $\phi_0$, $\phi_1$, $\phi_2$ appear in this order, hence we refer to~(\ref{ent-re1}) as relation of {\it entwining type} $(\phi_0, \phi_1, \phi_2)$ or by using just the subscripts, relation of entwining type $(0,1,2)$.

In Table~\ref{table4}, we present the entwining maps $S$, $T$, $U$ that correspond to the entwining relations (\ref{ent-re1})--(\ref{ent-re3}), where $R$ is any Yang--Baxter map. In what follows, we specify $R$ to be any member of the $H$-list\footnote{It is easy to show that the entwining maps associated with the $F$-list of quadrirational Yang--Baxter maps are $\text{(M\"ob)}^2$ equivalent to the corresponding to the $H$-list entwining maps. This is the reason that we present the entwining maps associated with the $H$-list only.} of quadrirational Yang--Baxter maps.

\begin{table}[!h]\centering
\caption{Entwining maps $S$, $T$, $U$ associated with a Yang--Baxter map $R$.} \label{table4}\vspace{1mm}
\begin{tabular}{c|c|c|c}
\hline
entwining type & $S_{12}$ & $T_{13}$ & $U_{23}$ \\ \hline
$(0,1,2)$ & $R_{12} $ & $R_{13}\phi_1 $ & $R_{23}\phi_2 $ \\ \hline
$(2,3,0)$ & $R_{12}\phi_2 $ & $ R_{13}\phi_3$ & $R_{23} $ \\ \hline
$(2,2,2)$ & $ R_{12}\phi_2$ & $ R_{13}$ & $ \phi_2 R_{23}\phi_2$ \\ \hline
\hline
\end{tabular}
\end{table}

\subsubsection[Entwining maps associated with the $H_{\rm I}$ Yang--Baxter map]{Entwining maps associated with the $\boldsymbol{H_{\rm I}}$ Yang--Baxter map}
The involutions $\phi$, $\psi$
\begin{gather*}
\phi\colon \ u\mapsto \frac{\alpha}{u}, \qquad \psi\colon \ u\mapsto \frac{u-\alpha}{u-1},
\end{gather*}
where $\alpha$ a complex parameter, are symmetries for the $H_{\rm I}$ map (see~\cite{pap3-2010}), since it holds
\begin{gather*}
\phi_1\phi_2R_{12}=R_{12}\phi_1\phi_2, \qquad \psi_1\psi_2R_{12}=R_{12}\psi_1\psi_2,
\end{gather*}
where $R_{12}$ is the $H_{\rm I}$ map acting on the $12$-coordinates and
\begin{alignat*}{3}
& \phi_1\colon \ (x_1,x_2)\mapsto (p_1/x_1,x_2), \qquad && \phi_2\colon \ (x_1,x_2)\mapsto (x_1,p_2/x_2), & \\
& \psi_1\colon \ (x_1,x_2)\mapsto ((x_1-p_1)/(x_1-1),x_2), \qquad && \psi_2\colon \ (x_1,x_2)\mapsto (x_1,(x_2-p_2)/(x_2-1)).&
 \end{alignat*}
Note that the symmetries $\phi$ and $\tau$ can be derived from our considerations (see Example~\ref{ex111}) since for $k=3$ it holds
\begin{alignat*}{3}
& H_1\phi_1\phi_2\phi_3=\frac{p_1p_2p_3}{H_1}, \qquad && H_2\phi_1\phi_2\phi_3=\frac{1}{H_2} ,& \\
& H_1\psi_1\psi_2\psi_3=H_2, \qquad &&H_2\psi_1\psi_2\psi_3=H_1.&
\end{alignat*}
\begin{Remark}
 By using similar arguments as in the proof of the Theorem \ref{thrm1}, entwining relations where the symmetries~$\phi$ and $\psi$ of the $H_{\rm I}$ map interlace do not exist, i.e., it does not exists for example any relation of entwining type $(\phi_i,\phi_j,\psi_k)$.
\end{Remark}

In Table~\ref{table5} we present the entwining maps associated with the $H_{\rm I}$ map which are generated by using the symmetries $\phi$ and $\psi$. In Table~\ref{table5} it appears the $H_{\rm I}$ map, the companion of the $H_{\rm I}$ map that is denoted as $cH_{\rm I}$, as well as ${\tilde cF_{\rm I}}$ which is the companion map of the map ${\tilde F_{\rm I}}$ that was derived in~\cite{pap3-2010}. We also have four novel maps which are not $\text{(M\"ob)}^2$ equivalent to $H_{\rm I}$, which we refer to as $\Phi_{\rm I}^a$, $\Phi_{\rm I}^b$, $\Psi_{\rm I}^a$ and $\Psi_{\rm I}^b$. In the proposition that follows we present their explicit form.

\begin{table}[!h]\centering
\caption{Left table: Entwining maps $S$, $T$, $U$ associated with $H_{\rm I}$ Yang--Baxter map using the sym\-met\-ry~$\phi$. Right table: Entwining maps $S$, $T$, $U$ associated with $H_{\rm I}$ Yang--Baxter map using the sym\-met\-ry~$\psi$.} \label{table5}\vspace{1mm}
\begin{tabular}{c|c|c|c}
\hline
entwining type & $S_{12}$ & $T_{13}$ & $U_{23}$ \\\hline
$(0,1,2)$ & $H_{\rm I} $ & $\Phi_{\rm I}^a $ & $\Phi_{\rm I}^a $ \tsep{2pt}\bsep{2pt}\\ \hline
$(2,3,0)$ & $\Phi_{\rm I}^b $ & $ \Phi_{\rm I}^b$ & $H_{\rm I} $ \tsep{2pt}\bsep{2pt}\\ \hline
$(2,2,2)$ & $ \Phi_{\rm I}^b$ & $ H_{\rm I}$ & $ cH_{\rm I}$ \tsep{2pt}\bsep{2pt}\\ \hline
\hline
\end{tabular}\qquad
\begin{tabular}{c|c|c|c}
\hline
entwining type & $S_{12}$ & $T_{13}$ & $U_{23}$ \\\hline
$(0,1,2)$ & $H_{\rm I} $ & $\Psi_{\rm I}^a $ & $\Psi_{\rm I}^a $ \tsep{2pt}\bsep{2pt}\\ \hline
$(2,3,0)$ & $\Psi_{\rm I}^b $ & $ \Psi_{\rm I}^b$ & $H_{\rm I} $\tsep{2pt}\bsep{2pt} \\ \hline
$(2,2,2)$ & $ \Psi_{\rm I}^b$ & $ H_{\rm I}$ & $ c{\tilde F_{\rm I}}$\tsep{2pt}\bsep{2pt} \\ \hline
\hline
\end{tabular}
\end{table}

\begin{Proposition}The following non-periodic\footnote{A non-periodic map cannot be equivalent by conjugation ($\text{(M\"ob)}^2$ equivalent) to a periodic map. Since the~$H_{\rm I}$ map is involutive, the maps presented in this proposition are not $\text{(M\"ob)}^2$ to the $H_{\rm I}$ map.} maps $(u,v)\mapsto (U,V)$, where
\begin{alignat*}{4}
& U=\alpha v Q,\quad && V=\frac{1}{u} Q^{-1} , \quad && Q=\frac{\beta-\alpha+u(1-\beta)+v(\alpha-1)}{\beta(1-\alpha)u-\alpha(1-\beta)v+(\alpha-\beta)uv} , & \tag*{$(\Phi_{\rm I}^a)$} \\
& U= \frac{1}{v} Q^{-1} , \quad && V=\beta u Q, \quad && Q=\frac{\beta-\alpha+u(1-\beta)+v(\alpha-1)}{\beta(1-\alpha)u-\alpha(1-\beta)v+(\alpha-\beta)uv} , & \tag*{$(\Phi_{\rm I}^b)$} \\
& U=v Q, \quad && V=\frac{u-\alpha}{u-1} Q^{-1}, \qquad && Q=\frac{\alpha(1-v)-\beta u+uv}{\beta(1-u)-\beta v+uv} , & \tag*{$(\Psi_{\rm I}^a)$} \\
& U=\frac{v-\beta}{v-1} Q , \qquad && V=u Q^{-1}, \quad && Q=\frac{\alpha(1-u-v)+uv}{\beta(1-u)-\alpha v+uv}, & \tag*{$(\Psi_{\rm I}^b)$}
\end{alignat*}
entwine with the $H_{\rm I}$ Yang--Baxter map according to the entwining relations of Table~{\rm \ref{table5}}.
\end{Proposition}

\subsubsection[Entwining maps associated with the $H_{\rm II}$ Yang--Baxter map]{Entwining maps associated with the $\boldsymbol{H_{\rm II}}$ Yang--Baxter map}
The invariants
\begin{gather*}
H_1=x_1+x_2+x_3,\qquad H_2=\frac{x_1-p_1}{x_1}\frac{x_2-p_2}{x_2}\frac{x_3-p_3}{x_3},
\end{gather*}
generate the maps $R_{ij}$, $i<j\in \{1,2,3\}$ which are exactly the $H_{\rm II}$ map acting on the $(ij)$-coordinates. Explicitly the $H_{\rm II}$ map reads
 \begin{gather*}
 U= v+\frac{(\alpha-\beta)uv}{\beta u+\alpha v-\alpha \beta}, \qquad V=u-\frac{(\alpha-\beta)uv}{\beta u+\alpha v-\alpha \beta}. \tag*{$(H_{\rm II})$}
\end{gather*}
 A symmetry of the $H_{\rm II}$ map is
$\phi\colon u\mapsto \alpha-u$, since it holds $\phi_1\phi_2R_{12}=R_{12}\phi_1\phi_2$, where $R_{12}$ is the $H_{\rm II}$ map acting on the $(12)$-coordinates and
\begin{gather*}
\phi_1\colon \ (x_1,x_2)\mapsto (p_1-x_1,x_2), \qquad \phi_2\colon \ (x_1,x_2)\mapsto (x_1,p_2-x_2).
\end{gather*}

\begin{table}[!h]\centering
\caption{Entwining maps $S$, $T$, $U$ associated with $H_{\rm II}$ Yang--Baxter map using the symmetry~$\phi$. } \label{table6}\vspace{1mm}
\begin{tabular}{c|c|c|c}
\hline
entwining type & $S_{12}$ & $T_{13}$ & $U_{23}$ \\ \hline
$(0,1,2)$ & $H_{\rm II} $ & $\Phi_{\rm II}^a $ & $\Phi_{\rm II}^a $ \tsep{2pt}\bsep{2pt}\\ \hline
$(2,3,0)$ & $\Phi_{\rm II}^b $ & $ \Phi_{\rm II}^b$ & $H_{\rm II} $ \tsep{2pt}\bsep{2pt}\\ \hline
$(2,2,2)$ & $ \Phi_{\rm II}^b$ & $ H_{\rm II}$ & $ cH_{\rm II}$ \tsep{2pt}\bsep{2pt}\\ \hline
\hline
\end{tabular}
\end{table}

\begin{Proposition}The following non-periodic maps $(u,v)\mapsto (U,V)$, where
\begin{alignat*}{3}
& U=\alpha v \frac{u-v+\beta-\alpha}{\beta u-\alpha v} , \qquad && V=\beta\frac{(\alpha-u)(u-v)}{\beta u-\alpha v} , & \tag*{$(\Phi_{\rm II}^a)$}\\
& U=\alpha\frac{(\beta-v)(u-v)}{\beta u-\alpha v} , \qquad && V=\beta u \frac{u-v+\beta-\alpha}{\beta u-\alpha v} , & \tag*{$(\Phi_{\rm II}^b)$}
\end{alignat*}
entwine with the $H_{\rm II}$ Yang--Baxter map according to the entwining relations of Table~{\rm \ref{table6}}.
\end{Proposition}

 The map $ cH_{\rm II}$ denotes the companion map of the $ H_{\rm II}$ map.

\subsubsection[Entwining maps associated with the $H_{\rm III}^A$ Yang--Baxter map]{Entwining maps associated with the $\boldsymbol{H_{\rm III}^A}$ Yang--Baxter map}

The invariants
\begin{gather*}
H_1=\frac{1}{x_1}+\frac{1}{x_2}+\frac{1}{x_3},\qquad H_2=p_1x_1+p_2x_2+p_3x_3,
\end{gather*}
generate the maps $R_{ij}$, $i<j\in \{1,2,3\}$ which are exactly the $H_{\rm III}^A$ map acting on the $(ij)$-coordinates. Explicitly the $H_{\rm III}^A$ map reads
 \begin{gather*}
U= \frac{v}{\alpha} \frac{\alpha u+\beta v}{u+v}, \qquad V=\frac{u}{\beta} \frac{\alpha u+\beta v}{u+v}. \tag*{$(H_{\rm III}^A)$}
\end{gather*}
 Two symmetries of the $H_{\rm III}^A$ map are
\begin{gather*}
\phi\colon \ u\mapsto \frac{1}{\alpha u},\qquad \psi\colon \ u\mapsto -u
\end{gather*}
since it holds
\begin{gather*}
\phi_1\phi_2R_{12}=R_{12}\phi_1\phi_2, \qquad \psi_1\psi_2R_{12}=R_{12}\psi_1\psi_2,
\end{gather*}
where $R_{12}$ is the $H_{\rm III}^A$ map acting on the $(12)$-coordinates and
\begin{alignat*}{3}
& \phi_1\colon \ (x_1,x_2)\mapsto \left(\frac{1}{p_1x_1},x_2\right), \qquad && \phi_2\colon \ (x_1,x_2)\mapsto \left(x_1,\frac{1}{p_2x_2}\right),& \\
& \psi_1\colon \ (x_1,x_2)\mapsto (-x_1,x_2), \qquad && \psi_2\colon \ (x_1,x_2)\mapsto (x_1,-x_2).&
\end{alignat*}
Note that the map $\phi_1 R_{12}\phi_2$ is exactly the $H_{\rm III}^B$ Yang--Baxter map.

\begin{Proposition}The following non-periodic maps $(u,v)\mapsto (U,V)$ where
\begin{alignat*}{3}
& U=v \frac{1+\beta u v}{1+\alpha u v} , \qquad && V=\frac{1}{\beta u}\frac{1+\beta u v}{1+\alpha u v}, & \tag*{$\big(\Phi_{{\rm III}^A}^a\big)$}\\
& U=\frac{1}{\alpha v}\frac{1+\alpha u v}{1+\beta u v} , \qquad && V=u \frac{1+\alpha u v}{1+\beta u v}, & \tag*{$\big(\Phi_{{\rm III}^A}^b\big)$}\\
& U=\frac{v}{\alpha}\frac{\alpha u-\beta v}{u-v} , \qquad && V=\frac{u}{\beta}\frac{\alpha u-\beta v}{v-u}, & \tag*{$\big(\Psi_{{\rm III}^A}^a\big)$}\\
& U=\frac{v}{\alpha}\frac{\alpha u-\beta v}{v-u} , \qquad && V=\frac{u}{\beta}\frac{\alpha u-\beta v}{u-v}, & \tag*{$\big(\Psi_{{\rm III}^A}^b\big)$}
\end{alignat*}
entwine with the $H_{\rm III}^A$ Yang--Baxter map according to the entwining relations of Table~{\rm \ref{table7}}.
\end{Proposition}

\begin{table}[!h]\centering
\caption{Left table: Entwining maps $S$, $T$, $U$ associated with $H_{\rm III}^A$ Yang--Baxter map using the sym\-met\-ry~$\phi$. Right table: Entwining maps $S$, $T$, $U$ associated with $H_{\rm III}^A$ Yang--Baxter map using the sym\-met\-ry~$\psi$.} \label{table7}\vspace{1mm}
\begin{tabular}{c|c|c|c}
\hline
entwining type & $S_{12}$ & $T_{13}$ & $U_{23}$ \\ \hline
$(0,1,2)$ & $H_{\rm III}^A $ & $\Phi_{{\rm III}^A}^a $ & $\Phi_{{\rm III}^A}^a $ \tsep{2pt}\bsep{2pt}\\ \hline
$(2,3,0)$ & $\Phi_{{\rm III}^A}^b $ & $ \Phi_{{\rm III}^A}^b$ & $H_{\rm III}^A $ \tsep{2pt}\bsep{2pt}\\ \hline
$(2,2,2)$ & $ \Phi_{{\rm III}^A}^b$ & $ H_{\rm III}^A$ & $ {\hat H_{\rm III}^A}$ \tsep{2pt}\bsep{2pt}\\ \hline
\hline
\end{tabular}\qquad
\begin{tabular}{c|c|c|c}
\hline
entwining type & $S_{12}$ & $T_{13}$ & $U_{23}$ \\ \hline
$(0,1,2)$ & $H_{\rm III}^A $ & $\Psi_{{\rm III}^A}^a $ & $\Psi_{{\rm III}^A}^a $ \tsep{2pt}\bsep{2pt}\\ \hline
$(2,3,0)$ & $\Psi_{{\rm III}^A}^b $ & $ \Psi_{{\rm III}^A}^b$ & $H_{\rm III}^A $ \tsep{2pt}\bsep{2pt}\\ \hline
$(2,2,2)$ & $ \Psi_{{\rm III}^A}^b$ & $ H_{\rm III}^A$ & $ cH_{\rm III}^A$ \tsep{2pt}\bsep{2pt}\\ \hline
\hline
\end{tabular}
\end{table}
The map $ cH_{\rm III}^A$ denotes the companion map of the $ H_{\rm III}^A$ map and with ${\hat H_{\rm III}^A}$ we denote a $\text{(M\"ob)}^2$ equivalent map to the~$ H_{\rm III}^A$.

\subsubsection[Entwining maps associated with the $H_{\rm III}^B$ Yang--Baxter map]{Entwining maps associated with the $\boldsymbol{H_{\rm III}^B}$ Yang--Baxter map}
The invariants that were derived in \cite{PKMN2,KaNie,KaNie:2018,PKMN3},
\begin{gather*}
H_1=x_1x_2x_3,\qquad H_2=p_1x_1+p_2x_2+p_3x_3+\frac{1}{x_1}+\frac{1}{x_2}+\frac{1}{x_3},
\end{gather*}
generate the maps $R_{ij}$, $i<j\in \{1,2,3\}$ which are exactly the $H_{\rm III}^B$ map acting on the $(ij)$-coordinates. Explicitly the $H_{\rm III}^B$ map reads
 \begin{gather*}
 U= v \frac{1+\beta u v}{1+\alpha u v}, \qquad V=u \frac{1+\alpha u v}{1+\beta u v} , \tag*{$\big(H_{\rm III}^B\big)$}
\end{gather*}
 The symmetries $\phi$, $\psi$ of the $H_{\rm III}^A$ map are symmetries of $H_{\rm III}^B$ as well.

\begin{Proposition}The following non-periodic maps $(u,v)\mapsto (U,V)$, where
\begin{alignat*}{3}
& U=\frac{v}{\alpha} \frac{\alpha u+\beta v}{u+v}, \qquad && V=\frac{1}{ u}\frac{u+v}{\alpha u +\beta v}, & \tag*{$\big(\Phi_{{\rm III}^B}^a\big)$}\\
& U=\frac{1}{v}\frac{u+v}{\alpha u+\beta v}, \qquad && V=\frac{u}{\beta} \frac{\alpha u +\beta v}{u+v}, & \tag*{$\big(\Phi_{{\rm III}^B}^b\big)$}\\
& U=v\frac{1-\beta u v}{1-\alpha u v}, \qquad && V=u\frac{1-\alpha uv}{-1+\beta uv}, & \tag*{$\big(\Psi_{{\rm III}^B}^a\big)$}\\
& U=v\frac{1-\beta uv}{-1+\alpha uv}, \qquad && V=u\frac{1-\alpha uv}{1-\beta uv}, & \tag*{$\big(\Psi_{{\rm III}^B}^b\big)$}
\end{alignat*}
entwine with the $H_{\rm III}^B$ Yang--Baxter map according to the entwining relations of Table~{\rm \ref{table8}}.
\end{Proposition}

\begin{table}[!h]\centering
\caption{Left table: Entwining maps $S$, $T$, $U$ associated with $H_{\rm III}^B$ Yang--Baxter map using the sym\-met\-ry~$\phi$. Right table: Entwining maps $S$, $T$, $U$ associated with $H_{\rm III}^B$ Yang--Baxter map using the sym\-metry~$\psi$.} \label{table8}\vspace{1mm}
\begin{tabular}{c|c|c|c}
\hline
entwining type & $S_{12}$ & $T_{13}$ & $U_{23}$ \\ \hline
$(0,1,2)$ & $H_{\rm III}^B $ & $\Phi_{{\rm III}^B}^a $ & $\Phi_{{\rm III}^B}^a $ \tsep{2pt}\bsep{2pt}\\ \hline
$(2,3,0)$ & $\Phi_{{\rm III}^B}^b $ & $ \Phi_{{\rm III}^B}^b$ & $H_{\rm III}^B $ \tsep{2pt}\bsep{2pt}\\ \hline
$(2,2,2)$ & $ \Phi_{{\rm III}^B}^b$ & $ H_{\rm III}^B$ & $ {\hat H_{\rm III}^B}$\tsep{2pt}\bsep{2pt} \\ \hline
\hline
\end{tabular}\qquad
\begin{tabular}{c|c|c|c}
\hline
entwining type & $S_{12}$ & $T_{13}$ & $U_{23}$ \\ \hline
$(0,1,2)$ & $H_{\rm III}^B $ & $\Psi_{{\rm III}^B}^a $ & $\Psi_{{\rm III}^B}^a $ \tsep{2pt}\bsep{2pt}\\ \hline
$(2,3,0)$ & $\Psi_{{\rm III}^B}^b $ & $ \Psi_{{\rm III}^B}^b$ & $H_{\rm III}^B $ \tsep{2pt}\bsep{2pt}\\ \hline
$(2,2,2)$ & $ \Psi_{{\rm III}^B}^b$ & $ H_{\rm III}^B$ & $ {\tilde H_{\rm III}^B}$ \tsep{2pt}\bsep{2pt}\\ \hline
\hline
\end{tabular}
\end{table}
The maps $ {\hat H_{\rm III}^B}$, ${\tilde H_{\rm III}^B}$ that appear in Table~\ref{table8}, are $\text{(M\"ob)}^2$ equivalent to the map $H_{\rm III}^B$. The map $cH_{\rm III}^B$ denotes the companion map of the $ H_{\rm III}^B$ map.

\subsubsection[Entwining maps associated with the $H_{\rm V}$ Yang--Baxter map]{Entwining maps associated with the $\boldsymbol{H_{\rm V}}$ Yang--Baxter map}
The invariants that were derived in \cite{PKMN2,KaNie,KaNie:2018,PKMN3},
\begin{gather*}
H_1=x_1+x_2+x_3,\qquad H_2=x_1^3+3p_1x_1+x_2^3+3p_2x_2+x_3^3+3p_3x_3,
\end{gather*}
generate the maps $R_{ij}$, $i<j\in \{1,2,3\}$ which are exactly the $H_{\rm V}$ map acting on the $(ij)$-coordinates. Explicitly the $H_{\rm V}$ map reads
 \begin{gather*}
 U= v - \frac{\alpha-\beta}{u+v}, \qquad V=u + \frac{\alpha-\beta}{u+v} . \tag*{$(H_{\rm V})$}
\end{gather*}
 The involution $ \psi\colon u\mapsto -u$ is a symmetry of the $H_{\rm V}$ map.

\begin{Proposition}
The following non-periodic maps $(u,v)\mapsto (U,V)$, where
\begin{alignat*}{3}
& U=v+\frac{\alpha-\beta}{u-v} , \qquad && V=-u-\frac{\alpha-\beta}{u-v}, & \tag*{$\big(\Psi_{\rm V}^a\big)$}\\
& U=-v-\frac{\alpha-\beta}{u-v} , \qquad && V=u+\frac{\alpha-\beta}{u-v}, & \tag*{$\big(\Psi_{\rm V}^b\big)$}
\end{alignat*}
entwine with the $H_{\rm V}$ Yang--Baxter map according to the entwining relations of Table~{\rm \ref{table9}}.
\end{Proposition}

\begin{table}[!h]\centering
\caption{Entwining maps $S$, $T$, $U$ associated with $H_{\rm V}$ Yang--Baxter map using the symmetry $\psi$.} \label{table9}\vspace{1mm}
\begin{tabular}{c|c|c|c}
\hline
entwining type & $S_{12}$ & $T_{13}$ & $U_{23}$ \\ \hline
$(0,1,2)$ & $H_{\rm V} $ & $\Psi_{\rm V}^a $ & $\Psi_{\rm V}^a $ \tsep{2pt}\bsep{2pt}\\ \hline
$(2,3,0)$ & $\Psi_{\rm V}^b $ & $ \Psi_{\rm V}^b$ & $H_{\rm V} $ \tsep{2pt}\bsep{2pt}\\ \hline
$(2,2,2)$ & $ \Psi_{\rm V}^b$ & $ H_{\rm V}$ & $ cH_{\rm V}$ \tsep{2pt}\bsep{2pt}\\ \hline
\hline
\end{tabular}
\end{table}
 The map $ cH_{\rm V}$ denotes the companion map of the $ H_{\rm V}$ map.

\section{Transfer maps} \label{Section5}

The notion of {\it transfer maps} associated with Yang--Baxter maps was introduced by Veselov in~\cite{Veselov:2003}. In~\cite{Veselov:2007} dynamical aspects of the latter were discussed.
The transfer maps associated with any reversible Yang--Baxter map are defined as
\begin{gather*}
T_i^{(k)}=R_{i i+k-1} R_{i i+k-2}\cdots R_{i i+1}, \qquad i\in\{1,\ldots, k\},
\end{gather*}
where the indices are considered modulo $k$.
 There is:
\begin{gather*}
T_i^{(k)} T_j^{(k)}=T_j^{(k)} T_i^{(k)}, \qquad T_1^{(k)} T_2^{(k)}\cdots T_k^{(k)}={\rm id}.
\end{gather*}
For example for $k=4$ we have
$T_1^{(4)}=R_{14}R_{13}R_{12}$, $T_2^{(4)}=R_{12}R_{24}R_{23}$, $T_3^{(4)}=R_{23}R_{13}R_{34}$ and $T_4^{(4)}=R_{34}R_{24}R_{14}$.
\begin{Proposition}
For the transfer maps $T_i^{(k)}$ associated with the maps $R_{ij}^{{\bf p}_{ij}}$ of the Proposi\-tions~{\rm \ref{prop2}}, {\rm \ref{prop3}}, {\rm \ref{prop4}}, it holds:
\begin{enumerate}\itemsep=0pt
\item[$1)$] they preserve the invariants $H_1$, $H_2$, presented in the Propositions~{\rm \ref{prop2}}, {\rm \ref{prop3}}, {\rm \ref{prop4}},
\item[$2)$] for $k=2n+1$ they preserve the measures given in the Propositions~{\rm \ref{prop2}}, {\rm \ref{prop3}}, {\rm \ref{prop4}},
\item[$3)$] for $k=2n$ they anti-preserve the measures given in the Propositions~{\rm \ref{prop2}}, {\rm \ref{prop3}}, {\rm \ref{prop4}},
\item[$4)$] they possess Lax pairs,
\item[$5)$] for generic values of the parameter sets ${\bf p}_{ij}$, are equivalent by conjugation to the transfer maps associated with $H_{\rm I}$, $H_{\rm II}$ and $H_{\rm III}^A$ Yang--Baxter maps respectively,
\item[$6)$] for non-generic values of the parameter sets ${\bf p}_{ij}$, we have novel transfer maps.
\end{enumerate}
\end{Proposition}

\begin{proof}The statements $(1)$--$(3)$ have already been proven (see Propositions~\ref{prop1}, \ref{prop2}, \ref{prop3}, \ref{prop4}). As for the statement~$(4)$, one can construct a Lax matrix for the Yang--Baxter map~$R$ fol\-lo\-wing~\cite{Veselov:2003b}. Then the Lax equations associated with the transfer maps~$T_i^{(k)}$, correspond to certain factorizations of the monodromy matrix (see~\cite{Veselov:2003}).

We will show the statement $(5)$ for the transfer maps associated with $R_{ij}^{{\bf p}_{ij}}$ of Proposition~\ref{prop2} and for $k=4$. The proof for arbitrary $k$ follows by induction.
In Proposition~\ref{prop2} it was shown that these maps are $\text{(M\"ob)}^2$ equivalent to the $H_{\rm I}$ map. Let us denote as $\nu_l$ the maps defined by the cross-ratios
 \begin{gather*}
\operatorname{CR}[x_l,a_l/b_l,c_l/d_l,A_l/B_l]=\operatorname{CR}[y_l,0,1,\infty],\qquad l=1,\ldots, 4,
\end{gather*}
and as $\mu_l$ the maps defined by
 \begin{gather*}
\operatorname{CR}[x_l,c_l/d_l,a_l/b_l,C_l/D_l]=\operatorname{CR}[y_l,\infty,1,0], \qquad l=1,\ldots, 4.
\end{gather*}
Then the maps $\tilde R_{ij}^{{\bf p}_{ij}}$, where $\tilde R_{ij}^{{\bf p}_{ij}}=\mu_j^{-1}\mu_i^{-1}R_{ij}^{{\bf p}_{ij}}\mu_i\mu_j$ are exactly the $H_{\rm I}$ map acting on the $(ij)$-coordinates (see Proposition~\ref{prop2}). For the transfer map $\tilde T_1^{(4)}$ associated with $\tilde R_{ij}^{{\bf p}_{ij}}$, there is
\begin{align}
\tilde T_1^{(4)}& =\tilde R_{14}\tilde R_{13}\tilde R_{12}=\big(\nu_1^{-1}\mu_4^{-1}R_{14}\nu_1\mu_4\big) \big(\nu_1^{-1}\mu_3^{-1}R_{13}\nu_1\mu_3\big) \big(\nu_1^{-1}\mu_3^{-1}R_{13}\nu_1\mu_2\big)\nonumber\\
& =\mu_4^{-1}\mu_3^{-1}\mu_2^{-1}\nu_1^{-1}R_{14}R_{13}R_{12}\nu_1\mu_2\mu_3\mu_4=\mu_4^{-1}\mu_3^{-1}\mu_2^{-1}\nu_1^{-1}T_1^{(4)}\nu_1\mu_2\mu_3\mu_4.\label{t-conj}
\end{align}
Note that we have omitted the parameter sets ${\bf p}_{ij}$ that the maps depends on for simplicity.

$(6)$. For non-generic choice of the parameter sets ${\bf p}_{ij}$, the conjugation equivalence~(\ref{t-conj}) does not holds.
\end{proof}

\subsection{On a re-factorisation of the transfer maps}
First, let us introduce some maps. With $ \pi_{ij}$ we denote the transpositions
\begin{gather*}
\pi_{ij}\colon \ (x_1,\ldots, x_k;{\bf p}_1,\ldots, {\bf p}_k)\mapsto (X_1,\ldots, X_k;{\bf P}_1,\ldots, {\bf P}_k),
\\
 X_l=x_l, \qquad {\bf P}_l={\bf p}_l \qquad \forall\, l\neq i,j, \qquad X_i=x_j, \qquad X_j=x_i, \qquad {\bf P}_i={\bf p}_j, \qquad {\bf P}_j={\bf p}_i.
\end{gather*}
and with $\pi_0$ we denote the following $k$-periodic map
\begin{gather*}
\pi_0\colon \ (x_1,\ldots, x_k;{\bf p}_1,\ldots, {\bf p}_k)\mapsto (X_1,\ldots, X_k;{\bf P}_1,\ldots, {\bf P}_k),
\\
 X_l=x_{l+1},\qquad {\bf P}_l={\bf p}_{l+1}, \qquad \forall\, l\in\{1,\ldots,k\}, \quad \mbox{modulo} \ k.
\end{gather*}
\begin{Remark} \label{rem5.1}
Note that $\pi_0=\pi_{12}\pi_{13}\cdots\pi_{1k}$ and the maps $\pi_0$, $\pi_{ij}$ $\forall\, i,j\in\{1,\ldots,k\}$, preserve the invariants $H_1$, $H_2$ of the Propositions \ref{prop2}, \ref{prop3}, \ref{prop4}. Moreover, the maps $S_i:=\pi_{i i+1}R_{i i+1}$, $i\in \{1,\ldots, k\}$, also preserve the invariants $H_1$, $H_2$. The following relations holds
\begin{gather*}
S_i^2=(S_iS_{i+1})^3=\pi_0^k={\rm id}, \qquad (S_iS_{j})^2={\rm id}, \qquad |i-j|>1, \qquad S_i\pi_0=\pi_0S_{i+1}.
\end{gather*}
The group $g=\langle \pi_0,S_1,S_2,\ldots, S_k\rangle $ generated by these maps provides a bi-rational realization of the extended Weyl group of type $A_{k-1}^{(1)}$.
\end{Remark}

\begin{Proposition} \label{prop-tr}
The transfer maps $T_i^{(k)}$ of a Yang--Baxter map $R$, coincide with the $(k-1)$-iteration of the maps
\begin{gather*}
t_i^{(k)}:=\pi_0 \pi_{i i+1}R_{i i+1}^{{\bf p}_{i i+1}}=\pi_0S_i.
\end{gather*}
We refer to the maps $t_i^{(k)}$ as the {\it extended transfer maps} associated with the Yang--Baxter map~$R$.
 \end{Proposition}

 \begin{proof}
It is enough to show that the $(k-1)$-iteration of the map $t_1^{(k)}$ coincides with $T_1^{(k)}$. For small values of $k$, this can be proven by direct calculation. In-order to complete the proof, it is enough to show that for arbitrary $k$ the maps $T_1^{(k)}$ and $\big(t_1^{(k)}\big)^{k-1}$ share the same Lax equation.

 Let $L(x,{\bf p};\boldsymbol{\lambda})$ the Lax matrix associated with the Yang--Baxter map $R$.
The Lax equation associated with the transfer map $T_1^{(k)}=R_{1 k}^{{\bf p}_{1 k}}R_{1 k-1}^{{\bf p}_{1 k-1}}\cdots R_{1 2}^{{\bf p}_{12}}$ reads
\begin{gather}
L(x_k,{\bf p}_k;\boldsymbol{\lambda})L(x_{k-1},{\bf p}_{k-1};\boldsymbol{\lambda})\cdots L(x_2,{\bf p}_2;\boldsymbol{\lambda}) L(x_1,{\bf p}_1;\boldsymbol{\lambda})\nonumber\\
\qquad {}= L(X_1,{\bf p}_1;\boldsymbol{\lambda})L(X_k,{\bf p}_k;\boldsymbol{\lambda})L(X_{k-1},{\bf p}_{k-1};\boldsymbol{\lambda})\cdots L(X_2,{\bf p}_2;\boldsymbol{\lambda}).\label{lax0}
\end{gather}
Since
\begin{gather*}
\pi_{1 2}R_{12}^{{\bf p}_{12}}\colon \ L(x_k,{\bf p}_k;\boldsymbol{\lambda})L(x_{k-1},{\bf p}_{k-1};\boldsymbol{\lambda})\cdots L(x_2,{\bf p}_2;\boldsymbol{\lambda}) L(x_1,{\bf p}_1;\boldsymbol{\lambda})\\
\qquad {} \mapsto L(x_k,{\bf p}_k;\boldsymbol{\lambda})L(x_{k-1},{\bf p}_{k-1};\boldsymbol{\lambda})\cdots L(x_2,{\bf p}_2;\boldsymbol{\lambda}) L(x_1,{\bf p}_1;\boldsymbol{\lambda}),
\end{gather*}
and
\begin{gather*}
\pi_0\colon \ L(x_k,{\bf p}_k;\boldsymbol{\lambda})L(x_{k-1},{\bf p}_{k-1};\boldsymbol{\lambda})\cdots L(x_2,{\bf p}_2;\boldsymbol{\lambda}) L(x_1,{\bf p}_1;\boldsymbol{\lambda})\\
\qquad {} \mapsto L(x_1,{\bf p}_1;\boldsymbol{\lambda})L(x_{k},{\bf p}_{k};\boldsymbol{\lambda})\cdots L(x_3,{\bf p}_3;\boldsymbol{\lambda}) L(x_2,{\bf p}_2;\boldsymbol{\lambda}),
\end{gather*}
there is
\begin{gather*}
t_1^{(k)}\colon \ L(x_k,{\bf p}_k;\boldsymbol{\lambda})L(x_{k-1},{\bf p}_{k-1};\boldsymbol{\lambda})\cdots L(x_2,{\bf p}_2;\boldsymbol{\lambda}) L(x_1,{\bf p}_1;\boldsymbol{\lambda})\\
\qquad{} \mapsto L(x_1,{\bf p}_1;\boldsymbol{\lambda})L(x_{k},{\bf p}_{k};\boldsymbol{\lambda})\cdots L(x_3,{\bf p}_3;\boldsymbol{\lambda}) L(x_2,{\bf p}_2;\boldsymbol{\lambda}).
\end{gather*}
So the map $t_1^{(k)}$ has the following Lax equation
\begin{gather*}
L(x_k,{\bf p}_k;\boldsymbol{\lambda})L(x_{k-1},{\bf p}_{k-1};\boldsymbol{\lambda})\cdots L(x_2,{\bf p}_2;\boldsymbol{\lambda}) L(x_1,{\bf p}_1;\boldsymbol{\lambda}) \\
 \qquad{} =L(X_1,{\bf P}_1;\boldsymbol{\lambda})L(X_k,{\bf P}_k;\boldsymbol{\lambda})L(X_{k-1},{\bf P}_{k-1};\boldsymbol{\lambda})\cdots L(X_2,{\bf P}_2;\boldsymbol{\lambda}).
\end{gather*}
But the map $t_1^{(k)}$ acts on the parameter sets ${\bf p}_i$ as follows
\begin{gather*}
t_1^{(k)} \colon \ ({\bf p}_1,\ldots, {\bf p}_k)\mapsto ({\bf P}_1,\ldots, {\bf P}_k),
\end{gather*}
where
\begin{gather*} {\bf P}_1={\bf p}_1, \qquad {\bf P}_k={\bf p}_2 \qquad \text{and} \qquad \forall\, i\neq 1,k \qquad {\bf P}_i={\bf p}_{i+1},
\end{gather*}
that is periodic with period $k-1$, so the Lax equation of the map $\big(t_1^{(k)}\big)^{k-1}$ is exactly~(\ref{lax0}), i.e., the Lax equation of~$T_1^{(k)}$.
\end{proof}

\begin{Theorem}The maps $t_i^{(k)}$ satisfy the relations
\begin{alignat*}{4}
& \big(t_i^{(k)} t_{i+1}^{(k)}\big)^{k/2}={\rm id},\qquad && t_1^{(k)}t_2^{(k)}\cdots t_k^{(k)}={\rm id},\qquad && k \quad \text{even},& \\
& \big(t_i^{(k)} t_{i+i}^{(k)}\big)^{k}={\rm id}, \qquad && \big(t_1^{(k)}t_2^{(k)}\cdots t_k^{(k)}\big)^2={\rm id},\qquad && k \quad \text{odd}.&
\end{alignat*}
\end{Theorem}
\begin{proof}
Let us first prove that $t_1^{(k)}t_2^{(k)}\cdots t_k^{(k)}={\rm id}$ for $k=2m$ even. There is
\begin{gather*}
t_1^{(2m)}t_2^{(2m)}\cdots t_{2m}^{(2m)}=\pi_0S_1\pi_0S_2\cdots \pi_0S_{2m},
\end{gather*}
where we have the composition of $m$ expressions of the form $\pi_0S_i\pi_0S_{i+1}$, and for each one of them (using Remark~\ref{rem5.1}) it holds
$\pi_0S_i\pi_0S_{i+1}=\pi_0S_i^2\pi_0=\pi_0^2$. So
\begin{gather*}
t_1^{(2m)}t_2^{(2m)}\cdots t_{2m}^{(2m)}= \underbrace{\pi_0^2 \pi_0^2\cdots \pi_0^2}_{\text{$m$-times}}=\pi_0^{2m}={\rm id}.
\end{gather*}
Let us now prove that $\big(t_i^{(k)} t_{i+1}^{(k)}\big)^{k/2}={\rm id}$. We have
\begin{gather*}
\big(t_i^{(k)} t_{i+1}^{(k)}\big)^{k/2}=\big(t_i^{(k)} t_{i+1}^{(k)}\big)^{m}=\big(\pi_0S_i\pi_0S_{i+1}\big)^m=\big(\pi_0^2S_{i+1}^2\big)^m=\pi_0^{2m}={\rm id}.
\end{gather*}

For $k=2m+1$ odd, we have
\begin{gather*}
\big(t_i^{(k)} t_{i+1}^{(k)}\big)^{k}=(t_i^{(k)} t_{i+1}^{(k)})^{2m+1}=\big(\pi_0^2S_{i+1}^2\big)^{2m+1}=\big(\pi_0^{2m+1}\big)^2={\rm id}.
\end{gather*}
Also,
\begin{align*}
\big(t_1^{(2m+1)}t_2^{(2m+1)}\cdots t_{2m+1}^{(2m+1)}\big)^2& =\big(t_1^{(2m+1)}t_2^{(2m+1)}\cdots t_{2m}^{(2m+1)}\pi_0S_{2m+1}\big)^2\\
& =\big(\pi_0^{2m+1}S_{2m+1}\big)^2=S_{2m+1}^2={\rm id},
\end{align*}
where we have used the fact that
\begin{gather*}
t_1^{(2m+1)}t_2^{(2m+1)}\cdots t_{2m}^{(2m+1)}= \underbrace{\pi_0^2 \pi_0^2\cdots \pi_0^2}_{\text{$m$-times}}=\pi_0^{2m}.\tag*{\qed}
\end{gather*}\renewcommand{\qed}{}
\end{proof}

\begin{Remark}Note that for $k$ odd, it holds the more general condition
\begin{gather*}
\big(t_i^{(k)} t_{j}^{(k)}\big)^{k}={\rm id}, \qquad i\neq j.
\end{gather*}
\end{Remark}

\subsection[$k$-point recurrences associated with the transfer maps of the $H$-list of quadrirational Yang--Baxter maps]{$\boldsymbol{k}$-point recurrences associated with the transfer maps\\ of the $\boldsymbol{H}$-list of quadrirational Yang--Baxter maps}

We refer to the extended transfer maps $t_i^{(k)}$ that correspond to the $H_{\rm I}$, $H_{\rm II}$, $H_{\rm III}^A$, $H_{\rm III}^B$ and $H_{\rm V}$ Yang--Baxter maps respectively as $ t_i^{H_{\rm I}(k)}$, $t_i^{H_{\rm II}(k)}$, $t_i^{H_{\rm III}^A(k)}$, $t_i^{H_{\rm III}^B(k)}$ and $t_i^{H_{\rm V}(k)}$.

Here, we associate $k$-point recurrences with the maps $ t_i^{H_{\rm I}(k)}$, $t_i^{H_{\rm II}(k)}$, $t_i^{H_{\rm III}^A(k)}$, $t_i^{H_{\rm III}^B(k)}$ and $t_i^{H_{\rm V}(k)}$. Let us first introduce the shift operator $T$ as follows
\begin{gather*}
T^0\colon \ x(n)\mapsto x(n),\qquad T^1\colon \ x(n)\mapsto x(n+1), \qquad T^l\colon \ x(n)\mapsto x(n+l), \\ T^{-l}\colon \ x(n)\mapsto x(n-l),\qquad n,l \in \mathbb{Z}.
\end{gather*}

The maps $ t_2^{H_{\rm I}(k)}$, $t_2^{H_{\rm II}(k)}$, $t_2^{H_{\rm III}^A(k)}$, $t_2^{H_{\rm III}^B(k)}$ and $t_2^{H_{\rm V}(k)}$, explicitly read
\begin{gather*}
(x_1,\ldots, x_k;p_1,\ldots, p_k) \mapsto (T x_1,\ldots, T x_k;T p_1,\ldots, T p_k),
\end{gather*}
where
\begin{gather*}
T x_1=x_2\frac{p_3(1-p_2)+(p_2-p_3)x_3+(p_3-1)x_2x_3}{p_2(1-p_3)+(p_3-p_2)x_2+(p_2-1)x_2x_3}, \qquad T p_1=p_3, \qquad T x_i=x_{i+1}, \\
 T x_2=x_3\frac{p_2(1-p_3)+(p_3-p_2)x_2+(p_2-1)x_2x_3}{p_3(1-p_2)+(p_2-p_3)x_3+(p_3-1)x_2x_3} , \qquad T p_2=p_2, \qquad T p_i=p_{i+1}, \tag*{$\big(t_2^{H_{\rm I}(k)}\big)$}\\
 T x_1=p_3x_2\frac{x_2+x_3-p_2}{p_3x_2+p_2x_3-p_2p_3}, \qquad T p_1=p_3, \qquad T x_i=x_{i+1}, \\
 T x_2=p_2x_3\frac{x_2+x_3-p_3}{p_3x_2+p_2x_3-p_2p_3} , \qquad T p_2=p_2, \qquad T p_i=p_{i+1}, \tag*{$\big(t_2^{H_{\rm II}(k)}\big)$}\\
T x_1=\frac{x_2}{p_3}\frac{p_2x_2+p_3x_3}{x_2+x_3} , \qquad T p_1=p_3, \qquad T x_i=x_{i+1}, \\
 T x_2=\frac{x_3}{p_2}\frac{p_2x_2+p_3x_3}{x_2+x_3} , \qquad T p_2=p_2, \qquad T p_i=p_{i+1}, \tag*{$\big(t_2^{H_{\rm III}^A(k)}\big)$}\\
 T x_1=x_2\frac{1+p_2x_2x_3}{1+p_3x_2x_3}, \qquad T p_1=p_3, \qquad T x_i=x_{i+1}, \\
T x_2=x_3\frac{1+p_3x_2x_3}{1+p_2x_2x_3}, \qquad T p_2=p_2, \qquad T p_i=p_{i+1}, \tag*{$\big(t_2^{H_{\rm III}^B(k)}\big)$}\\
T x_1=x_2-\frac{p_3-p_2}{x_2+x_3} , \qquad T p_1=p_3, \qquad T x_i=x_{i+1}, \\
T x_2=x_3+\frac{p_3-p_2}{x_2+x_3} , \qquad T p_2=p_2, \qquad T p_i=p_{i+1},\tag*{$\big(t_2^{H_{\rm V}(k)}\big)$}
\end{gather*}
with $i=3,4,\ldots, k$ and $T x_k=x_1$, $T p_k=p_1$. Moreover, not just $t_2^{(k)}$, but all the maps $t_i^{(k)}$, $i=1,2,\ldots, k$, preserve the invariants in separated variables (see Table~\ref{table999})\footnote{The invariants in separated variables that appear in Table~\ref{table999}, were firstly introduced, in a different context, in~\cite{PKMN2,KaNie,KaNie:2018,PKMN3}. Note that the invariants $H_1$, $H_2$ for $t_i^{H_{\rm III}^A(k)}$ were also given in~\cite{Kouloukas:2018}.} and they anti-preserve the measures $m_i=n^i d^{i+1}$ where $n^i$, $d^i$ the numerator and the denominator respectively of the invariants $H_i$, $i=1,2$. Additional invariant can be constructed though the Lax formulation (see the proof of Proposition~\ref{prop-tr}).

\begin{table}[!h]\centering
\caption{Invariants in separated variables for the maps $ t_i^{H_{\rm I}(k)}$, $t_i^{H_{\rm II}(k)}$, $t_i^{H_{\rm III}^A(k)}$, $t_i^{H_{\rm III}^B(k)}$ and $t_i^{H_{\rm V}(k)}$.} \label{table999}\vspace{1mm}
\begin{tabular}{c|c|c}
\hline
map & $H_1$ & $H_2$ \\ \hline
$t_i^{H_{\rm I}(k)}$ & ${\dis \prod_{i=1}^k p_i x_i } $ & ${\dis \prod_{i=1}^k\frac{x_i-p_i}{x_i-1}\frac{1}{p_i-1} } $ \\ [3mm]
$t_i^{H_{\rm II}(k)}$ & ${\dis \sum_{i=1}^k x_i+p_i} $ & $ {\dis \prod_{i=1}^k\frac{x_i-p_i}{p_i x_i} } $ \\ [3mm]
$t_i^{H_{\rm III}^A(k)}$ & ${\dis \sum_{i=1}^k \frac{1}{x_i}+\frac{1}{p_i}} $ & ${\dis \sum_{i=1}^k p_ix_i }$ \\ [3mm]
$t_i^{H_{\rm III}^B(k)}$ & $ {\dis \prod_{i=1}^k p_i x_i }$ & $ {\dis \sum_{i=1}^k \frac{1}{x_i} + p_ix_i+\frac{1}{p_i} }$ \\ [3mm]
$t_i^{H_{\rm V}(k)}$ & ${\dis \sum_{i=1}^k x_i+p_i } $ & $ {\dis \sum_{i=1}^k x_i^3+3p_ix_i+p_i^3} $ \\
\hline
\end{tabular}
\end{table}

Now we show how a $k$-point recurrence can be associated with the map $t_2^{H_{\rm V}(k)}$. Recall that the map $t_2^{H_{\rm V}(k)}$ reads
\begin{gather*}
t_2^{H_{\rm V}(k)}\colon \ (x_1,\ldots, x_k;p_1,\ldots, p_k) \mapsto (T x_1,\ldots, T x_k;T p_1,\ldots, T p_k),
\end{gather*}
where
\begin{gather*}
T x_1=x_2-\frac{p_3-p_2}{x_2+x_3}, \qquad T x_2=x_3+\frac{p_3-p_2}{x_2+x_3},\qquad T x_i=x_{i+1}, \\ T p_1=p_3, \qquad T p_2=p_2, \qquad T p_i=p_{i+1}, \qquad i=3,\ldots, k,
\end{gather*}
and the indices are considered modulo $k$. Clearly we have, $x_3=T^{2-k} x_1$, $p_3=T^{2-k} p_1$. So we obtain
\begin{gather}
T x_1=x_2-\frac{T^{2-k} p_1-p_2}{x_2+T^{2-k} x_1}, \qquad T x_2=T^{2-k} x_1+\frac{T^{2-k} p_1-p_2}{x_2+T^{2-k} x_1},\nonumber\\
 T^{k-1} p_1=p_1, \qquad T p_2=p_2.\label{kh51}
\end{gather}
Adding the first two equations from above we get the following invariance condition\footnote{This condition is a consequence of the fact that the $t_i^{H_{\rm V}(k)}$ preserves the invariant $H_1=\sum\limits_{i=1}^k x_i$. Such a~condition exists for the remaining extended transfer maps associated with the Yang--Baxter maps of the $H$-list. The latter enable us to write~$t_2^{(k)}$ maps as $k$-point recurrences. }
\begin{gather} \label{invc1}
\big(T^1-T^{2-k}\big)x_1=\big(T^0-T^1\big)x_2.
\end{gather}
So it is guaranteed the existence of a potential function $f$ such that
\begin{gather*}
x_1=c+\big(T^0-T^1\big)f, \qquad x_2=c+\big(T^1-T^{2-k}\big)f, \qquad \mbox{where} \qquad c ={\rm const}.
\end{gather*}
In terms of $f$, (\ref{kh51}) becomes the following $(k+1)$-point recurrence
\begin{gather} \label{kh52}
\big(T^2-T^{2-k}\big)f=\frac{-p_2+T^{2-k} p_1}{2c+\big(T-T^{3-k}\big)f},\qquad T^{k-1} p_1=p_1, \qquad T p_2=p_2.
\end{gather}
In terms of a new variable $h$ defined as $h:=\lambda +\big(T^1-T^0\big)f$, there is,
\begin{gather*}
\big(T^2-T^{2-k}\big)f=-\lambda k+\sum_{i=2-k}^1 T^i h,\qquad \big(T-T^{3-k}\big)f=\lambda (2-k)+\sum_{i=3-k}^0 T^i h,
\end{gather*}
 so~(\ref{kh52}) becomes the $k$-point recurrence
\begin{gather*} 
\frac{2ck}{2-k}+\sum_{i=2-k}^1 T^i h= \frac{-p_2+T^{2-k} p_1}{\sum\limits_{i=3-k}^0 T^i h} , \qquad T^{k-1} p_1=p_1, \qquad T p_2=p_2,
\end{gather*}
where we chose $\lambda=\frac{2c}{k-2}$ to simplify the formulae.

\begin{table}[!h]\centering
\caption{The invariance conditions~(\ref{invc1}) and the potential functions $f$ for the maps $ t_2^{H_{\rm I}(k)}$, $t_2^{H_{\rm II}(k)}$, $t_2^{H_{\rm III}^A(k)}$, $t_2^{H_{\rm III}^B(k)}$ and $t_2^{H_{\rm V}(k)}$.} \label{table10}\vspace{1mm}
\begin{tabular}{c|c|c}
\hline
map & invariance condition & potential function $f$ \\ \hline
$t_2^{H_{\rm I}(k)}$ & ${\dis \frac{T x_1}{T^{2-k} x_1}=\frac{T^0 x_2}{T x_2} } $\tsep{8pt} & $x_1=c\dfrac{T^0 f}{T f}$, $x_2 =c\dfrac{T f}{T^{2-k} f} $ \bsep{5pt}\\
$t_2^{H_{\rm II}(k)}$ & $ \big(T-T^{2-k}\big)x_1=\big(T^0-T\big)x_2 $ & $ x_1=c+\big(T^0-T\big)f$, $x_2=c+\big(T-T^{2-k}\big)f $ \bsep{5pt}\\
$t_2^{H_{\rm III}^A(k)}$ & ${\dis \big(T-T^{2-k}\big)\frac{1}{x_1}=\big(T^0-T\big)\frac{1}{x_2} } $ & $\dfrac{1}{x_1}=\dfrac{1}{c}+\big(T^0-T\big)f$, $\dfrac{1}{x_2}=\dfrac{1}{c}+\big(T-T^{2-k}\big)f$ \bsep{5pt}\\
$t_2^{H_{\rm III}^B(k)}$ & $ {\dis \frac{T x_1}{T^{2-k} x_1}=\frac{T^0 x_2}{T x_2} }$ & $ x_1=c\dfrac{T^0 f}{T f}$, $x_2 =c\dfrac{T f}{T^{2-k} f}$ \bsep{5pt}\\
$t_2^{H_{\rm V}(k)}$ & $ \big(T-T^{2-k}\big)x_1=\big(T^0-T\big)x_2 $ & $ x_1=c+\big(T^0-T\big)f$, $x_2=c+\big(T-T^{2-k}\big)f $ \bsep{2pt}\\
\hline
\end{tabular}
\end{table}

\begin{Proposition} \label{prop52}The following $(k+1)$-point recurrences corresponds to the extended transfer map $t_2^{(k)}$ associated with $H_{\rm I}$, $H_{\rm II}$, $H_{\rm III}^A$, $H_{\rm III}^B$ and $H_{\rm V}$ Yang--Baxter maps respectively. We refer to these $(k+1)$-point recurrences respectively as $rt_2^{H_{\rm I}(k)}$, $rt_2^{H_{\rm I}(k)}$, $rt_2^{H_{\rm III}^A(k)}$, $rt_2^{H_{\rm III}^B(k)}$ and $rt_2^{H_{\rm V}(k)}$
\begin{gather*}
 \frac{T^2 f}{T^{2-k}f} = \frac{ p_2\big({-}1+T^{2-k} p_1\big)+c\big(p_2-T^{2-k} p_1\big)\dfrac{T^1f}{T^{2-k} f}+c^2(1-p_2)\dfrac{T^1f}{T^{3-k}f} }
{\big(T^{2-k} p_1\big)(p_2-1)+c\big({-}p_2+T^{2-k} p_1\big)\dfrac{T^{2-k}f}{T^{3-k} f}+c^2\big(1-T^{2-k} p_1\big)\dfrac{T^1f}{T^{3-k}f}},
\tag*{$\big(rt_2^{H_{\rm I}(k)}\big)$}\\
\frac{c+\big(T^1-T^2\big)f}{c+\big(T^1-T^{2-k}\big)f}\\
{}=\frac{\big(2c-p_2+\big(T^1-T^{3-k}\big)f\big) T^{2-k} p_1}{-p_2T^{2-k} p_1+c\big(p_2+T^{2-k} p_1\big)+\big(T^{2-k} p_1\big)\big(T^1-T^{2-k}\big)f+p_2\big(T^{2-k}-T^{3-k}\big)f},
\!\!\!\!\!\tag*{$\big(rt_2^{H_{\rm II}(k)}\big)$}\\
 \frac{c+\big(T^1-T^2\big)f}{c+\big(T^1-T^{2-k}\big)f}=
 \frac{2c+\big(T^1-T^{3-k}\big)f}{c+\big(T^1-T^{2-k}\big)f+\dfrac{p_2}{T^{2-k}p_1}\big(c+\big(T^{2-k}-T^{3-k}\big)f\big)},
\tag*{$\big(rt_2^{H_{\rm III}^A(k)}\big)$}\\
\frac{T^2 f}{T^{2-k}f}=\frac{T^{3-k}f+c^2\big(T^{2-k}p_1\big)T^1f}{T^{3-k}f+c^2p_2T^1f},
\tag*{$\big(rt_2^{H_{\rm III}^B(k)}\big)$}\\
(T^2-T^{2-k})f=\frac{-p_2+T^{2-k}p_1}{2c+\big(T-T^{3-k}f\big)}.
\tag*{$\big(rt_2^{H_{\rm V}(k)}\big)$}
\end{gather*}
For each recurrence presented above we have that the parameters vary as follows: $T p_2=p_2$, $T^{k-1} p_1=p_1$. So $p_2$ is constant and $p_1$ is periodic with period $k-1$.
\end{Proposition}
Note that the recurrences $rt_2^{H_{\rm I}(k)}$ and $rt_2^{H_{\rm III}^B(k)}$ are bilinear. Some members of $rt_2^{H_{\rm I}(k)}$ and $rt_2^{H_{\rm III}^B(k)}$, for specific choices of the parameters $c$, $p_2$ and of the function~$p_1$, are expected to exhibit the {\it Laurent property} \cite{Fomin1:2002,Fomin2:2002,Fordy:2014}.

\begin{table}[!h]\centering
\caption{Definition of the variables $h$ associated with the recurrences of Proposition~\ref{prop52}.} \label{table13}\vspace{1mm}
\begin{tabular}{c|c|c}
 \hline
 recurrence & variable $h$ & a choice for $\lambda$ \\ \hline
 $rt_2^{H_{\rm I}(k)}$ & $h:=\lambda \dfrac{T f}{T^0 f}$\tsep{7pt} & $\lambda=\dfrac{1}{c} $\bsep{7pt} \\
 $rt_2^{H_{\rm II}(k)}$ & $h:=\lambda+\big(T-T^0\big)f$ & $\lambda=\dfrac{2c}{k-2} $\bsep{7pt} \\
 $rt_2^{H_{\rm III}^A(k)}$ & $h:=\lambda+\big(T-T^0\big)f$ & $\lambda=\dfrac{2c}{k-2}$\bsep{7pt} \\
 $rt_2^{H_{\rm III}^B(k)}$ & $h:=\lambda \dfrac{T f}{T^0 f}$ & $\lambda=\dfrac{1}{c} $\bsep{7pt} \\
 $rt_2^{H_{\rm V}(k)}$ & $h:=\lambda+\big(T-T^0\big)f$ & $\lambda=\dfrac{2c}{k-2} $\bsep{7pt} \\
 \hline
\end{tabular}
\end{table}

\begin{Corollary} \label{cor1}
The $(k+1)$-point recurrences $rt_2^{H_{\rm I}(k)}$, $rt_2^{H_{\rm II}(k)}$, $rt_2^{H_{\rm III}^A(k)}$, $rt_2^{H_{\rm III}^B(k)}$ and $rt_2^{H_{\rm V}(k)}$, in terms of the corresponding variables~$h$ defined in Table~{\rm \ref{table13}}, get the form of the following $k$-point recurrences
\begin{gather*}
 \prod_{i=3-k}^1T^i h=\frac{{\dis c^{-k}p_2\big(T^{2-k} p_1-1\big)+\big(p_2-T^{2-k} p_1\big)\prod_{i=2-k}^0 T^i h+(1-p_2)\prod_{i=3-k}^0 T^i h}}
{{\dis T^{2-k} p_1-p_2+T^{2-k} p_1 (p_2-1)T^{2-k} h+c^k\big(1-T^{2-k} p_1\big)\prod_{i=3-k}^0 T^i h }} ,
\tag*{$\big(\hat{r}t_2^{H_{\rm I}(k)}\big)$}\\
 p_2\left(\frac{c k}{k-2}-T h\right)\left(\frac{c k}{k-2}-T^{2-k}p_1-T^{2-k}h\right) \\
\qquad{} =T^{2-k}p_1\left(\frac{c k}{k-2}-\sum_{i=2-k}^0 T^i h\right)\left(\frac{c k}{k-2}+p_2-\sum_{i=3-k}^1 T^i h\right) ,
\tag*{$\big(\hat{r}t_2^{H_{\rm II}(k)}\big)$}\\
 p_2\left(\frac{c k}{k-2}-T h\right)\left(\frac{c k}{k-2}-T^{2-k}h\right) \\
\qquad{} = T^{2-k}p_1\left(\frac{c k}{k-2}-\sum_{i=2-k}^0 T^i h\right)\left(\frac{c k}{k-2}-\sum_{i=3-k}^1 T^i h\right) ,
\tag*{$\big(\hat{r}t_2^{H_{\rm III}^A(k)}\big)$}\\
 \prod_{i=2-k}^1 T^i h = \frac{{\dis c^{-k}+T^{2-k}p_1\prod_{i=3-k}^0 T^i h}}{{\dis 1+c^kp_2\prod_{i=3-k}^0 T^i h}},
\tag*{$\big(\hat{r}t_2^{H_{\rm III}^B(k)}\big)$}\\
 -\frac{2 c k}{k-2}+\sum_{i=2-k}^1 T^i h = \frac{{\dis -p_2+T^{2-k}p_1}}{{\dis \sum_{i=3-k}^0 T^i h}}
\tag*{$\big(\hat{r}t_2^{H_{\rm V}(k)}\big)$}
\end{gather*}
and for each recurrence presented above we have that the parameters vary as follows: $T p_2=p_2$, $T^{k-1} p_1=p_1$. So $p_2$ is constant and $p_1$ is periodic with period $k-1$.
\end{Corollary}

Note that the $(k+1)$-point recurrences of Proposition~\ref{prop52}, as well as the corresponding $k$-point ones introduced in Corollary~\ref{cor1} are non-autonomous. This is due to the fact that~$p_1$ varies periodically $(T^{k-1}p_1=p_1)$. The non-autonomous terms that will be introduced by integrating the relation $T^{k-1}p_1=p_1$ are periodic though. Proper de-autonomization for the recurrences~$\hat{r}t_2^{H_{\rm V}(k)}$ and~$\hat{r}t_2^{H_{\rm III}^B(k)}$ will be introduced in what follows.

\subsubsection[The recurrences $\hat{r}t_i^{H_{\rm V}(k)}$ and discrete Painlev\'e equations]{The recurrences $\boldsymbol{\hat{r}t_i^{H_{\rm V}(k)}}$ and discrete Painlev\'e equations}
The dressing chain for the KdV equation \cite{Veselov:1993}, reads
\begin{gather} \label{dr1}
(g_{i+1}+g_i)_t=g_{i+1}^2-g_i^2+p_{i+1}-p_i.
\end{gather}
The recurrences $\hat{r}t_i^{H_{\rm V}(k)}$, serve as its discretisations. Actually they are exactly the $(k-1)$-roots of the discretisations presented in~\cite{adler-1993}. So, $\hat{r}t_i^{H_{\rm V}(k)}$ corresponds to Liouville integrable maps.

Since the dressing chain~(\ref{dr1}) leads to Painlev\'e equations $P_{\rm IV}$ and $P_{\rm V}$ and their higher order analogues~\cite{Veselov:1993}, the recurrences $\hat{r}t_i^{H_{\rm V}(k)}$ (after proper de-autonomisation) can be considered as their discrete counter-parts and/or the B\"{a}cklund transformations of the higher order~$P_{\rm IV}$ and~$P_{\rm V}$ Painlev\'e equations.

A proper de-autonomisation of $\hat{r}t_2^{H_{\rm V}(k)}$ is achieved by breaking the periodicity of the $p_1$ assuming that $T^{k-1}p_1=p_1+(k-1)a$, where $a$ constant. This de-autonomisation is proper since the resulting non-autonomous discrete system preserves the same Poisson structure\footnote{The Poisson structures associated with the dressing chain for the KdV equation were first derived in~\cite{Veselov:1993}, see also~\cite{Evripidou:2017}.} as the autonomous one. So we obtain the following hierarchy of discrete Painlev\'e equations
\begin{gather} \label{dpI}
 -\frac{2 c k}{k-2}+\sum_{i=2-k}^1 T^i h=\frac{-p_2+T^{2-k}p_1}{\sum\limits_{i=3-k}^0 T^i h}, \qquad T p_2=p_2,\qquad T^{k-1} p_1=p_1+(k-1)a.
\end{gather}
For $k=3$, (\ref{dpI}) reads
\begin{gather*}
-6c+T h+h+T^{-1} h=\frac{-p_2+T^{-1}p_1}{h},\qquad T p_2=p_2,\qquad T^{2} p_1=p_1+2a.
\end{gather*}
So $p_2$ is constant and $p_1=b_0+b_1(-1)^n+an$, with $b_0$, $b_1$, $a$ constants. We can choose $-p_2+b_0=b$ constant, hence we obtain the following discrete Painlev\'e equation which serves as B\"{a}cklund transformation of $P_{\rm IV}$~\cite{Noumi1998}
\begin{gather} \label{dp1}
-6c+T h+h+T^{-1} h=\frac{b+b_1(-1)^n+an}{h}, \qquad n\in \mathbb{Z}.
\end{gather}
For $k=4$, (\ref{dpI}) reads
\begin{gather*}
-4c+T^{-2} h+h+T^{-1} h+h+Th=\frac{-p_2+T^{-2}p_1}{h+T^{-1} h},\qquad T p_2=p_2, \qquad T^{3} p_1=p_1+3a.
\end{gather*}
If we define a new variable $w$ as $w:=h+T^{-1} h$, then we obtain the following discrete Painlev\'e equation which serves as B\"{a}cklund transformation of~$P_{\rm V}$
\begin{gather*}
-4c+T^{-1} w+T w=\frac{-p_2+T^{-2}p_1}{w}, \qquad T p_2=p_2, \qquad T^{3} p_1=p_1+3a.
\end{gather*}

So for $k$ odd (\ref{dpI}) serves as B\"{a}cklund transformation for the higher order analogues of $P_{\rm IV}$ and for~$k$ even~(\ref{dpI}) serves as B\"{a}cklund transformation for the higher order analogues of $P_{\rm V}$. Note that in~\cite{Noumi1998}, B\"{a}cklund transformation for the higher order analogues of $P_{\rm IV}$ and~$P_{\rm V}$ were given in terms of continued fractions. We can recover the form of discrete Painlev\'e equations introduced in~\cite{Noumi1998} by making use of the alternating terms that appear in~(\ref{dpI}). For example for $k=3$, the term $(1)^n$ that appears in~(\ref{dp1}), suggests the introduction of the variables \mbox{$y(m):=h(2n)$}, $z(m):=h(2n+1)$. Then (\ref{dp1}) takes to form of the {\it second discrete Painlev\'e equation~$dP_{\rm II}$}
\begin{gather*}
 y+z+T^{-1}z=\frac{b_0+b_1+am}{y}, \qquad T y+y+z=\frac{b_0-b_1+am}{z},\qquad m\in\mathbb{Z}.
\end{gather*}

\subsubsection[The recurrences $\hat{r}t_i^{H_{\rm III}^B(k)}$ and discrete Painlev\'e equations]{The recurrences $\boldsymbol{\hat{r}t_i^{H_{\rm III}^B(k)}}$ and discrete Painlev\'e equations}

As we plan to show in our future work, the recurrences $\hat{r}t_i^{H_{\rm III}^B(k)}$ serves as Liouville integrable discretisations of the following chain introduced in~\cite{Adler2006bb}
\begin{gather*}
(g_i+g_{i+1})_t=2(p_i\cosh g_i-p_{i+1}\cosh g_{i+1}).
\end{gather*}

A proper de-autonomisation of $\hat{r}t_2^{H_{\rm III}^B(k)}$ is achieved by breaking the periodicity of the $p_1$ in a~way that the non-autonomous system preserves the same Poisson structure as the autonomous one. This is achieved by imposing that $T^{k-1}p_1=p_1 a^{k-1}$, where $a$ constant. So we obtain the following hierarchy of discrete Painlev\'e equations
\begin{gather}\label{qpIII}
 \prod_{i=2-k}^1 T^i h = \frac{{\dis c^{-k}+T^{2-k}p_1\prod_{i=3-k}^0 T^i h}}{{\dis 1+c^kp_2\prod_{i=3-k}^0 T^i h}} ,\qquad T p_2=p_2, \qquad T^{k-1} p_1=p_1 a^{k-1}.
\end{gather}
For $k=3$, (\ref{qpIII}) reads
\begin{gather*}
T h T^{-1} h=\frac{1}{h}\frac{c^{-3}+hT^{-1}p_1}{1+c^3p_2 h}, \qquad T p_2=p_2, \qquad T^{2} p_1=p_1 a^2.
\end{gather*}
So $p_2$ is constant and $p_1=b_0a^n+b_1(-a)^n$, with $b_0$, $b_1$, $a$ constants. Hence we obtain the $q-P_{\rm I}\big(A_6^{(1)}\big)$ discrete Painlev\'e equation (see~\cite{Sakai}).
For $k=4$, (\ref{qpIII}) reads
\begin{gather*}
T h T^0 h T^{-1} h T^{-2}h=\frac{c^{-4}+h T^{-1} h T^{-2}p_1}{1+c^4p_2 h T^{-1} h},\qquad T p_2=p_2, \qquad T^{3} p_1=p_1 a^3.
\end{gather*}
If we define a new variable $w$ as $w:=hT^{-1} h$, then we obtain the $q-P_{\rm II}\big(A_5^{(1)}\big)$ discrete Painlev\'e equation (see~\cite{Sakai})
\begin{gather*}
T w T^{-1} w=\frac{c^{-4}+w T^{-2}p_1}{1+c^4p_2 w},\qquad T p_2=p_2, \qquad T^{3} p_1=p_1 a^3.
\end{gather*}
The Lax pair associated with the hierarchy~(\ref{qpIII}) first appeared in~\cite{Hay:2007}.
\begin{Remark}As for the recurrences $\hat{r}t_i^{H_{\rm III}^A(k)}$, $\hat{r}t_i^{H_{\rm II}(k)}$, one could consider $T^{k-1} p_1=p_1+(k-1)a$ and for $\hat{r}t_i^{H_{\rm I}(k)}$ $T^{k-1} p_1=p_1 a^{k-1}$, in order to de-autonomise them. We anticipate that this is a~proper de-autonomisation, although we have no proof yet. The finding of the Poisson structures that the latter recurrences we anticipate that preserve, will sort this issue out.
\end{Remark}

\begin{Remark}As a final remark, we note that the $k$-point recurrences associated with the extended transfer maps of the Yang--Baxter map $F_{\rm V}$, are exactly the same as the $k$-point recurrences associated with the extended transfer maps of the Yang--Baxter map $H_{\rm V}$ which (one of them) were presented in Corollary~\ref{cor1}. Since the $(k-1)$-iteration of the extended transfer maps of any Yang--Baxter map coincides with its transfer maps, we conclude that the dynamics of the transfer maps of the Yang--Baxter maps $F_{\rm V}$ and $H_{\rm V}$, are the same. The same holds true for the transfer maps associated with the Yang--Baxter maps $F_{\rm III}$ and $H_{\rm III}^A$. As for the remaining members of the $F$ and the $H$ lists of Yang--Baxter maps, further investigation is required in order to prove the equivalence of their transfer dynamics.
\end{Remark}

 \section{Conclusions}\label{Section6}
In Section \ref{Section2} we have presented a family of maps in $k$ variables which preserve $2$ rational invariants of a specific form. One could mimic the procedures introduced in \cite{Kassotakis:2006} to obtain rational maps in $k$ variables which preserve $m$ rational invariants where $m<k$. For example, there are ${2k\choose k}$ rational maps $(x_1,\ldots, x_k, y_1,\ldots, y_k)\mapsto (X_1,\ldots, X_k, Y_1,\ldots, Y_k)$ which preserve $k$ invariants of the form:
\begin{gather} \label{cont}
H_i=\frac{\alpha_ix_ix_{i+1}+\beta_ix_i+\gamma_ix_{i+1}+\delta_i}{\kappa_ix_ix_{i+1}+\lambda_ix_i+\mu_ix_{i+1}+\nu_i},\qquad i=1,2,\ldots,k,
\end{gather}
where the indices are considered modulo $k$ and $\alpha_i$, $\beta_i$, $\kappa_i$, $\lambda_i$, etc. are given functions of the variables $y_i$, $y_{i+1}$.

If separability of variables on the invariants is imposed, then higher rank analogues of the Yang--Baxter maps of Propositions \ref{prop2}, \ref{prop3} and~\ref{prop4} are expected. Moreover, solutions of the functional tetrahedron equation \cite{Kashaev:1996,Korepanov:1998,Korepanov:1995,Sergeev:1998}, or even of higher simplex equations \cite{Dimakis:2015, Maillet:1989,Maillet:1989b} are anticipated. For example if we consider the following, different than~(\ref{cont}), choice of invariants:
\begin{gather*}
H_1=\sum_{i=1}^6 x_i,\qquad H_2=\frac{x_1 x_4 x_6}{x_3}, \qquad H_3=x_2x_3x_4x_5,
\end{gather*}
then the involutions $R_{123}$, $R_{145}$, $R_{246}$, and $R_{356}$, preserve $H_i$, $i=1,2,3$ and satisfy the functional tetrahedron equation
\begin{gather*}
R_{123}R_{145} R_{246} R_{356}=R_{356} R_{246} R_{145} R_{123}.
\end{gather*}
They are exactly the Hirota's map \cite{Kashaev:1996,Korepanov:1998,Sergeev:1998}, i.e., the map $R\colon (u,v,w)\mapsto (U,V,W)$, where
\begin{gather*}
U=\frac{uv}{u+w},\qquad V=u+w,\qquad W=\frac{vw}{u+w},
\end{gather*}
acting on $(123)$, $(145)$, $(246)$ and $(356)$ coordinates respectively. For the involution $\phi\colon u\mapsto -u$, it holds $\phi_1\phi_2\phi_3 R_{123}=R_{123}\phi_1\phi_2\phi_3$. So $\phi$ is a symmetry of the Hirota's map $R$ and it can be easily proven that the following entwining relation holds
\begin{gather*}
R_{123}\phi_3R_{145}\phi_5 R_{246}\phi_6 R_{356}=R_{356} R_{246}\phi_6 R_{145}\phi_5 R_{123}\phi_3.
\end{gather*}
Hence we have obtained a solution of the following entwining functional tetrahedron relation
\begin{gather*}
S_{123}S_{145} S_{246} T_{356}=T_{356} S_{246} S_{145} S_{123},
\end{gather*}
where $T$ is the Hirota's map acting on the $(356)$ coordinates and $S\colon (u,v,w)\mapsto (U,V,W)$ a~non-periodic map where
\begin{gather*}
U=\frac{uv}{u-w},\qquad V=u-w,\qquad W=-\frac{vw}{u-w}.
\end{gather*}
The complete set of entwining relations and maps associated with the Hirota's map as well as with the Hirota--Miwa's map, will be considered elsewhere.

In Section~\ref{Section4}, we considered two methods to obtain entwining maps. The first method uses degeneracy arguments and produces entwining maps associated with the $H_{\rm I}$, $H_{\rm II}$ and $H_{\rm III}^A$ Yang--Baxter maps. The entwining maps of this method belongs to different subclasses than the $[2:2]$ subclass of maps that the $H_{\rm I}$, $H_{\rm II}$ and $H_{\rm III}^A$ Yang--Baxter maps belongs to so they are not $\text{(M\"ob)}^2$ equivalent to the latter. The outcomes of the second method are non-periodic\footnote{The non-periodicity assures that these entwining maps are not $\text{(M\"ob)}^2$ equivalent with the corresponding maps of the $H$-list.} entwining maps of subclass $[2:2]$ associated with the whole $H$-list. The fact that the entwining maps which were presented in this Section preserve two invariants in separated variables, enable us to introduce appropriate potentials (as shown in \cite{PKMN2,KaNie,PKMN3}) to obtain integrable lattice equations. Actually we obtain integrable triplets of lattice equations (in some cases even correspondences). Note that integrable triplets of lattice equations were systematically derived in~\cite{Boll:2011} and more recently in~\cite{Hietarinta:2018}. We plan to consider the integrable triplets of lattice equations derived from entwining maps, elsewhere.

In Section~\ref{Section6}, we have proved that the transfer maps associated with the $H$ list of Yang--Baxter maps can be considered as the $(k-1)$-iteration of some maps of simpler form. As a~consequence of this re-factorisation we have obtained $(k+1)$-point (see Proposition~\ref{prop52}) and $k$-point (see Corollary~\ref{cor1}) alternating recurrences which can be considered as alternating versions of some hierarchies of discrete Painlev\'e equations. Moreover, the autonomous versions of some of the $k$-point recurrences presented in Corollary~\ref{cor1}, can be obtained by {\it periodic reductions}~\cite{Papageorgiou:1990} (cf.~\cite{Frank-Nalini-book}) of integrable lattice equations. Here we have obtained alternating $k$-point recurrences from Yang--Baxter maps without performing periodic reductions. Hence, our results might be compared/extended to the novel and independent frameworks introduced in \cite{Atkinson:2012i,Atkinson:2018} and \cite{Joshi:2015,Joshi:2016a}, where by using symmetry arguments, integrable lattice equations and discrete Painlev\'e equations of $2$nd order were linked.

\appendix

\section[The $F$-list and the $H$-list of quadrirational Yang--Baxter maps]{The $\boldsymbol{F}$-list and the $\boldsymbol{H}$-list\\ of quadrirational Yang--Baxter maps} \label{app1}
The Yang--Baxter maps $R$ of the $F$ and the $H$-list, explicitly read
\begin{gather*}
R\colon \ \mathbb{CP}^1\times \mathbb{CP}^1\ni (u,v)\mapsto (U,V)\in \mathbb{CP}^1\times \mathbb{CP}^1,
\\
U=\alpha v P,\qquad V=\beta u P,
\qquad P=\frac{(1-\beta)u+\beta-\alpha+(\alpha-1)v}{\beta(1-\alpha)u+(\alpha-\beta)uv+\alpha (\beta-1) v}, \tag*{$(F_{\rm I})$} \\
U=\alpha v P,\qquad V=\beta u P,
\qquad P=\frac{u-v+\beta-\alpha}{\beta u-\alpha v},\tag*{$(F_{\rm II})$} \\
U=\frac{v}{\alpha} P,\qquad V=\frac{u}{\beta} P,
\qquad P=\frac{\alpha u-\beta v}{u-v}, \tag*{$(F_{\rm III})$}\\
U=v P,\qquad V=u P,
\qquad P=1+\frac{\beta -\alpha}{u-v},\tag*{$(F_{\rm IV})$} \\
U=v+ P,\qquad V=u+ P,
\qquad P=\frac{\alpha-\beta}{u-v},\tag*{$(F_{\rm V})$}\\
U=v Q,\qquad V=uQ^{-1},
\qquad Q=\frac{(\alpha-1)uv+(\beta-\alpha)u+\alpha(1-\beta)}{(\beta-1)uv+(\alpha-\beta)v+\beta(1-\alpha)},\tag*{$(H_{\rm I})$}\\
U=v+Q,\qquad V=u-Q,
\qquad Q=\frac{(\alpha-\beta)uv}{\beta u+\alpha v-\alpha \beta},\tag*{$(H_{\rm II})$} \\
U=\frac{v}{\alpha} Q,\qquad V=\frac{u}{\beta} Q,
\qquad Q=\frac{\alpha u+\beta v}{u+v},\tag*{$\big(H_{\rm III}^A\big)$}\\
U=v Q,\qquad V=u Q^{-1},
\qquad Q=\frac{1+\beta u v}{1+\alpha uv},\tag*{$\big(H_{\rm III}^B\big)$}\\
U=v-Q,\qquad V=u+Q,
\qquad Q=\frac{\alpha -\beta}{u+v}.\tag*{$(H_{\rm V})$}
\end{gather*}

The maps above are depending on $2$ complex parameters $\alpha$, $\beta$. The parameter $\alpha$ is associated with the first factor of the cartesian product
$\mathbb{CP}^1\times \mathbb{CP}^1$, whereas the parameter $\beta$ with the second factor.

\subsection*{Acknowledgements}

P.K.\ is grateful to Aristophanis Dimakis, Vassilios Papageorgiou and Anastasios Tongas, the organizers of the $4^{\rm th}$ Workshop on Mathematical Physics-Integrable Systems (November~30 -- December 1, 2018, Department of Mathematics, University of Patras, Patras, Greece), where this work was finalized. Also, P.K.\ is grateful to James Atkinson, Allan Fordy, Nalini Joshi, Frank Nijhoff and to Pol Vanhaecke for very fruitful discussions on the subject, as well as to Maciej Nieszporski for the endless discussions towards the answer to the great question of integrable systems, Yang--Baxter maps and everything.

\pdfbookmark[1]{References}{ref}
\LastPageEnding


\begin{thebibliography}{99}
\footnotesize\itemsep=0pt

\bibitem{adler-1993}
Adler V.E., Recuttings of polygons, \href{https://doi.org/10.1007/BF01085984}{\textit{Funct. Anal. Appl.}} \textbf{27}
 (1993), 141--143.

\bibitem{Adler:1998}
Adler V.E., B\"{a}cklund transformation for the {K}richever--{N}ovikov
 equation, \href{https://doi.org/10.1155/S1073792898000014}{\textit{Int. Math. Res. Not.}} \textbf{1998} (1998), 1--4,
 \href{https://arxiv.org/abs/solv-int/9707015}{arXiv:solv-int/9707015}.

\bibitem{Adler:2006}
Adler V.E., On a class of third order mappings with two rational invariants,
 \href{https://arxiv.org/abs/nlin.SI/0606056}{arXiv:nlin.SI/0606056}.

\bibitem{ABS-2003}
Adler V.E., Bobenko A.I., Suris Yu.B., Classification of integrable equations on
 quad-graphs. {T}he consistency approach, \href{https://doi.org/10.1007/s00220-002-0762-8}{\textit{Comm. Math. Phys.}}
 \textbf{233} (2003), 513--543, \href{https://arxiv.org/abs/nlin.SI/0202024}{arXiv:nlin.SI/0202024}.

\bibitem{ABSyb-2004}
Adler V.E., Bobenko A.I., Suris Yu.B., Geometry of {Y}ang--{B}axter maps:
 pencils of conics and quadrirational mappings, \href{https://doi.org/10.4310/CAG.2004.v12.n5.a1}{\textit{Comm. Anal. Geom.}}
 \textbf{12} (2004), 967--1007, \href{https://arxiv.org/abs/math.QA/0307009}{arXiv:math.QA/0307009}.

\bibitem{Adler2006bb}
Adler V.E., Shabat A.B., Dressing chain for the acoustic spectral problem,
 \href{https://doi.org/10.1007/s11232-006-0121-6}{\textit{Theoret. and Math. Phys.}} \textbf{149} (2006), 1324--1337,
 \href{https://arxiv.org/abs/nlin.SI/0604008}{arXiv:nlin.SI/0604008}.

\bibitem{Adler:1994}
Adler V.E., Yamilov R.I., Explicit auto-transformations of integrable chains,
 \href{https://doi.org/10.1088/0305-4470/27/2/030}{\textit{J.~Phys.~A: Math. Gen.}} \textbf{27} (1994), 477--492.

\bibitem{Atkinson:2012i}
Atkinson J., Idempotent biquadratics, {Y}ang--{B}axter maps and birational
 representations of {C}oxeter groups, \href{https://arxiv.org/abs/1301.4613}{arXiv:1301.4613}.

\bibitem{AtkNie}
Atkinson J., Nieszporski M., Multi-quadratic quad equations: integrable cases
 from a factorized-discriminant hypothesis, \href{https://doi.org/10.1093/imrn/rnt066}{\textit{Int. Math. Res. Not.}}
 \textbf{2014} (2014), 4215--4240, \href{https://arxiv.org/abs/1204.0638}{arXiv:1204.0638}.

\bibitem{Atkinson:2018}
Atkinson J., Yamada Y., Quadrirational {Y}ang--{B}axter maps and the elliptic
 {C}remona system, \href{https://arxiv.org/abs/1804.01794}{arXiv:1804.01794}.

\bibitem{baxter-1982}
Baxter R.J., Exactly solved models in statistical mechanics, Academic Press,
 Inc., London, 1982.

\bibitem{BAZHANOV2018509}
Bazhanov V.V., Sergeev S.M., Yang--{B}axter maps, discrete integrable equations
 and quantum groups, \href{https://doi.org/10.1016/j.nuclphysb.2017.11.017}{\textit{Nuclear Phys.~B}} \textbf{926} (2018), 509--543,
 \href{https://arxiv.org/abs/1501.06984}{arXiv:1501.06984}.

\bibitem{Boll:2011}
Boll R., Classification of 3{D} consistent quad-equations, \href{https://doi.org/10.1142/S1402925111001647}{\textit{J.~Nonlinear
 Math. Phys.}} \textbf{18} (2011), 337--365, \href{https://arxiv.org/abs/1009.4007}{arXiv:1009.4007}.

\bibitem{Santini:1991}
Bruschi M., Ragnisco O., Santini P.M., Tu G.Z., Integrable symplectic maps,
 \href{https://doi.org/10.1016/0167-2789(91)90149-4}{\textit{Phys.~D}} \textbf{49} (1991), 273--294.

\bibitem{Capel:2001}
Capel H.W., Sahadevan R., A new family of four-dimensional symplectic and
 integrable mappings, \href{https://doi.org/10.1016/S0378-4371(00)00314-9}{\textit{Phys.~A}} \textbf{289} (2001), 86--106.

\bibitem{Joshi:1990}
Cresswell C., Joshi N., The discrete first, second and thirty-fourth
 {P}ainlev\'{e} hierarchies, \href{https://doi.org/10.1088/0305-4470/32/4/009}{\textit{J.~Phys.~A: Math. Gen.}} \textbf{32}
 (1999), 655--669.

\bibitem{Dimakis:2015}
Dimakis A., M\"{u}ller-Hoissen F., Simplex and polygon equations,
 \href{https://doi.org/10.3842/SIGMA.2015.042}{\textit{SIGMA}} \textbf{11} (2015), 042, 49~pages, \href{https://arxiv.org/abs/1409.7855}{arXiv:1409.7855}.

\bibitem{Dimakis2018}
Dimakis A., M\"{u}ller-Hoissen F., Matrix {K}adomtsev--{P}etviashvili
 equation: tropical limit, {Y}ang--{B}axter and pentagon maps,
 \href{https://doi.org/10.1134/S0040577918080056}{\textit{Theoret. and Math. Phys.}} \textbf{196} (2018), 1164--1173,
 \href{https://arxiv.org/abs/1709.09848}{arXiv:1709.09848}.

\bibitem{Dimakis2018ii}
Dimakis A., M\"{u}ller-Hoissen F., Matrix {KP}: tropical limit and
 {Y}ang--{B}axter maps, \href{https://doi.org/10.1007/s11005-018-1127-3}{\textit{Lett. Math. Phys.}} \textbf{109} (2019),
 799--827, \href{https://arxiv.org/abs/1708.05694}{arXiv:1708.05694}.

\bibitem{Doliwa:2014}
Doliwa A., Non-commutative rational {Y}ang--{B}axter maps, \href{https://doi.org/10.1007/s11005-013-0669-7}{\textit{Lett. Math.
 Phys.}} \textbf{104} (2014), 299--309, \href{https://arxiv.org/abs/1308.2824}{arXiv:1308.2824}.

\bibitem{drinfeld-1992}
Drinfel'd V.G., On some unsolved problems in quantum group theory, in Quantum
 Groups ({L}eningrad, 1990), \textit{Lecture Notes in Math.}, Vol.~1510,
 \href{https://doi.org/10.1007/BFb0101175}{Springer}, Berlin, 1992, 1--8.

\bibitem{Duistermaat:2010}
Duistermaat J.J., Discrete integrable systems. QRT maps and elliptic surfaces,
 \textit{Springer Monographs in Mathematics}, \href{https://doi.org/10.1007/978-0-387-72923-7}{Springer}, New York, 2010.

\bibitem{et-2003}
Etingof P., Geometric crystals and set-theoretical solutions to the quantum
 {Y}ang--{B}axter equation, \href{https://doi.org/10.1081/AGB-120018516}{\textit{Comm. Algebra}} \textbf{31} (2003),
 1961--1973, \href{https://arxiv.org/abs/math.QA/0112278}{arXiv:math.QA/0112278}.

\bibitem{etingof-1999}
Etingof P., Schedler T., Soloviev A., Set-theoretical solutions to the quantum
 {Y}ang--{B}axter equation, \href{https://doi.org/10.1215/S0012-7094-99-10007-X}{\textit{Duke Math.~J.}} \textbf{100} (1999),
 169--209, \href{https://arxiv.org/abs/math.QA/9801047}{arXiv:math.QA/9801047}.

\bibitem{Evripidou:2017}
Evripidou C.A., Kassotakis P., Vanhaecke P., Integrable deformations of the
 {B}ogoyavlenskij--{I}toh {L}otka--{V}olterra systems, \href{https://doi.org/10.1134/S1560354717060090}{\textit{Regul. Chaotic
 Dyn.}} \textbf{22} (2017), 721--739, \href{https://arxiv.org/abs/1709.06763}{arXiv:1709.06763}.

\bibitem{Fomin1:2002}
Fomin S., Zelevinsky A., Cluster algebras. {I}.~{F}oundations, \href{https://doi.org/10.1090/S0894-0347-01-00385-X}{\textit{J.~Amer.
 Math. Soc.}} \textbf{15} (2002), 497--529, \href{https://arxiv.org/abs/math.RT/0104151}{arXiv:math.RT/0104151}.

\bibitem{Fomin2:2002}
Fomin S., Zelevinsky A., The {L}aurent phenomenon, \href{https://doi.org/10.1006/aama.2001.0770}{\textit{Adv. in Appl. Math.}}
 \textbf{28} (2002), 119--144, \href{https://arxiv.org/abs/math.CO/0104241}{arXiv:math.CO/0104241}.

\bibitem{Fordy:2014}
Fordy A.P., Hone A., Discrete integrable systems and {P}oisson algebras from
 cluster maps, \href{https://doi.org/10.1007/s00220-013-1867-y}{\textit{Comm. Math. Phys.}} \textbf{325} (2014), 527--584,
 \href{https://arxiv.org/abs/1207.6072}{arXiv:1207.6072}.

\bibitem{Kassotakis:2006}
Fordy A.P., Kassotakis P., Multidimensional maps of {QRT} type,
 \href{https://doi.org/10.1088/0305-4470/39/34/012}{\textit{J.~Phys.~A: Math. Gen.}} \textbf{39} (2006), 10773--10786.

\bibitem{Kassotakis:2013}
Fordy A.P., Kassotakis P., Integrable maps which preserve functions with
 symmetries, \href{https://doi.org/10.1088/1751-8113/46/20/205201}{\textit{J.~Phys.~A: Math. Theor.}} \textbf{46} (2013), 205201,
 12~pages, \href{https://arxiv.org/abs/1301.1927}{arXiv:1301.1927}.

\bibitem{Grahovski:2016}
Grahovski G.G., Konstantinou-Rizos S., Mikhailov A.V., Grassmann extensions of
 {Y}ang--{B}axter maps, \href{https://doi.org/10.1088/1751-8113/49/14/145202}{\textit{J.~Phys.~A: Math. Theor.}} \textbf{49} (2016),
 145202, 17~pages, \href{https://arxiv.org/abs/1510.06913}{arXiv:1510.06913}.

\bibitem{Hay:2007}
Hay M., Hierarchies of nonlinear integrable {$q$}-difference equations from
 series of {L}ax pairs, \href{https://doi.org/10.1088/1751-8113/40/34/005}{\textit{J.~Phys.~A: Math. Theor.}} \textbf{40} (2007),
 10457--10471.

\bibitem{Hietarinta:2018}
Hietarinta J., Search for {CAC}-integrable homogeneous quadratic triplets of
 quad equations and their classification by {BT} and {L}ax,
 \href{https://doi.org/10.1080/14029251.2019.1613047}{\textit{J.~Nonlinear Math. Phys.}} \textbf{26} (2019), 358--389,
 \href{https://arxiv.org/abs/1806.08511}{arXiv:1806.08511}.

\bibitem{Frank-Nalini-book}
Hietarinta J., Joshi N., Nijhoff F.W., Discrete systems and integrability,
\textit{Cambridge Texts in Applied Mathe\-ma\-tics}, \href{https://doi.org/10.1017/CBO9781107337411}{Cambridge University Press},
 Cambridge, 2016.

\bibitem{Iatrou:2003}
Iatrou A., Higher dimensional integrable mappings, \href{https://doi.org/10.1016/S0167-2789(03)00011-3}{\textit{Phys.~D}}
 \textbf{179} (2003), 229--253.

\bibitem{Jimbo-1989}
Jimbo M. (Editor), Yang--{B}axter equation in integrable systems, \textit{Adv.
 Ser. Math. Phys.}, Vol.~10, \href{https://doi.org/10.1142/1021}{World Sci. Publ.}, Teaneck, NJ, 1989.

\bibitem{Joshi:2018}
Joshi N., Kassotakis P., Re-factorising a {QRT} map, \href{https://arxiv.org/abs/1906.00501}{arXiv:1906.00501}.

\bibitem{Joshi:2015}
Joshi N., Nakazono N., Shi Y., Lattice equations arising from discrete
 {P}ainlev\'{e} systems. {I}.~{$(A_2 + A_1)^{(1)}$} and {$\big(A_1 +
 A_1^\prime\big)^{(1)}$} cases, \href{https://doi.org/10.1063/1.4931481}{\textit{J.~Math. Phys.}} \textbf{56} (2015),
 092705, 25~pages, \href{https://arxiv.org/abs/1401.7044}{arXiv:1401.7044}.

\bibitem{Joshi:2016a}
Joshi N., Nakazono N., Shi Y., Lattice equations arising from discrete
 {P}ainlev\'{e} systems: {II}.~{$A^{(1)}_4$} case, \href{https://doi.org/10.1088/1751-8113/49/49/495201}{\textit{J.~Phys.~A: Math.
 Theor.}} \textbf{49} (2016), 495201, 39~pages, \href{https://arxiv.org/abs/1603.09414}{arXiv:1603.09414}.

\bibitem{KNY-2002A}
Kajiwara K., Noumi M., Yamada Y., Discrete dynamical systems with
 {$W\big(A_{m-1}^{(1)}\times A_{n-1}^{(1)}\big)$} symmetry, \href{https://doi.org/10.1023/A:1016298925276}{\textit{Lett.
 Math. Phys.}} \textbf{60} (2002), 211--219, \href{https://arxiv.org/abs/nlin.SI/0106029}{arXiv:nlin.SI/0106029}.

\bibitem{Kashaev:1996}
Kashaev R.M., On discrete three-dimensional equations associated with the local
 {Y}ang--{B}axter relation, \href{https://doi.org/10.1007/BF01815521}{\textit{Lett. Math. Phys.}} \textbf{38} (1996),
 389--397, \href{https://arxiv.org/abs/solv-int/9512005}{arXiv:solv-int/9512005}.

\bibitem{Korepanov:1998}
Kashaev R.M., Korepanov I.G., Sergeev S.M., Functional tetrahedron equation,
 \href{https://doi.org/10.1007/BF02557179}{\textit{Theoret. and Math. Phys.}} \textbf{117} (1998), 1402--1413,
 \href{https://arxiv.org/abs/solv-int/9801015}{arXiv:solv-int/9801015}.

\bibitem{Kassotakis:phd}
Kassotakis P., The construction of discrete dynamical system, Ph.D.~Thesis,
 {U}niversity of Leeds, 2006.

\bibitem{PKMN2}
Kassotakis P., Nieszporski M., Families of integrable equations, \href{https://doi.org/10.3842/SIGMA.2011.100}{\textit{SIGMA}}
 \textbf{7} (2011), 100, 14~pages, \href{https://arxiv.org/abs/1106.0636}{arXiv:1106.0636}.

\bibitem{KaNie}
Kassotakis P., Nieszporski M., On non-multiaffine consistent around the cube
 lattice equations, \href{https://doi.org/10.1016/j.physleta.2012.10.009}{\textit{Phys. Lett.~A}} \textbf{376} (2012), 3135--3140,
 \href{https://arxiv.org/abs/1106.0435}{arXiv:1106.0435}.

\bibitem{KaNie:2017}
Kassotakis P., Nieszporski M., {$2^n$}-rational maps, \href{https://doi.org/10.1088/1751-8121/aa6dbd}{\textit{J.~Phys.~A: Math.
 Theor.}} \textbf{50} (2017), 21LT01, 9~pages, \href{https://arxiv.org/abs/1512.00771}{arXiv:1512.00771}.

\bibitem{KaNie:2018}
Kassotakis P., Nieszporski M., Difference systems in bond and face variables
 and non-potential versions of discrete integrable systems,
 \href{https://doi.org/10.1088/1751-8121/aad4c4}{\textit{J.~Phys.~A: Math. Theor.}} \textbf{51} (2018), 385203, 21~pages,
 \href{https://arxiv.org/abs/1710.11111}{arXiv:1710.11111}.

\bibitem{Rizos:2013}
Konstantinou-Rizos S., Mikhailov A.V., Darboux transformations, finite
 reduction groups and related {Y}ang--{B}axter maps, \href{https://doi.org/10.1088/1751-8113/46/42/425201}{\textit{J.~Phys.~A: Math.
 Theor.}} \textbf{46} (2013), 425201, 16~pages.

\bibitem{Korepanov:1995}
Korepanov I.G., Algebraic integrable dynamical systems, 2+1-dimensional models
 in wholly discrete space-time, and inhomogeneous models in 2-dimensional
 statistical physics, \href{https://arxiv.org/abs/solv-int/9506003}{arXiv:solv-int/9506003}.

\bibitem{Kouloukas:2018}
Kouloukas T.E., Relativistic collisions as {Y}ang--{B}axter maps, \href{https://doi.org/10.1016/j.physleta.2017.09.007}{\textit{Phys.
 Lett.~A}} \textbf{381} (2017), 3445--3449, \href{https://arxiv.org/abs/1706.06361}{arXiv:1706.06361}.

\bibitem{Kouloukas:2011}
Kouloukas T.E., Papageorgiou V.G., Entwining {Y}ang--{B}axter maps and
 integrable lattices, in Algebra, Geometry and Mathematical Physics,
\textit{Banach Center Publ.}, Vol.~93, \href{https://doi.org/10.4064/bc93-0-13}{Polish Acad. Sci. Inst. Math.}, Warsaw,
 2011, 163--175.

\bibitem{Maeda:1987}
Maeda S., Completely integrable symplectic mapping, \href{https://doi.org/10.3792/pjaa.63.198}{\textit{Proc. Japan Acad.
 Ser.~A Math. Sci.}} \textbf{63} (1987), 198--200.

\bibitem{Maillet:1989}
Maillet J.-M., Nijhoff F., Integrability for multidimensional lattice models,
 \href{https://doi.org/10.1016/0370-2693(89)91466-4}{\textit{Phys. Lett.~B}} \textbf{224} (1989), 389--396.

\bibitem{Maillet:1989b}
Maillet J.-M., Nijhoff F., The tetrahedron equation and the four-simplex
 equation, \href{https://doi.org/10.1016/0375-9601(89)90400-3}{\textit{Phys. Lett.~A}} \textbf{134} (1989), 221--228.

\bibitem{Haggar:1996}
McLachlan R.I., Quispel G.R.W., Generating functions for dynamical systems with
 symmetries, integrals, and differential invariants, \href{https://doi.org/10.1016/S0167-2789(97)00218-2}{\textit{Phys.~D}}
 \textbf{112} (1998), 298--309.

\bibitem{PKMN3}
Nieszporski M., Kassotakis P., Systems of difference equations on a vector
 valued function that admit a 3{D} vector space of scalar potentials, {i}n
 preparation.

\bibitem{Noumi1998}
Noumi M., Yamada Y., Affine {W}eyl groups, discrete dynamical systems and
 {P}ainlev\'{e} equations, \href{https://doi.org/10.1007/s002200050502}{\textit{Comm. Math. Phys.}} \textbf{199} (1998),
 281--295, \href{https://arxiv.org/abs/math.QA/9804132}{arXiv:math.QA/9804132}.

\bibitem{Papageorgiou:1990}
Papageorgiou V.G., Nijhoff F.W., Capel H.W., Integrable mappings and nonlinear
 integrable lattice equations, \href{https://doi.org/10.1016/0375-9601(90)90876-P}{\textit{Phys. Lett.~A}} \textbf{147} (1990),
 106--114.

\bibitem{pap3-2010}
Papageorgiou V.G., Suris Yu.B., Tongas A.G., Veselov A.P., On quadrirational
 {Y}ang--{B}axter maps, \href{https://doi.org/10.3842/SIGMA.2010.033}{\textit{SIGMA}} \textbf{6} (2010), 033, 9~pages,
 \href{https://arxiv.org/abs/0911.2895}{arXiv:0911.2895}.

\bibitem{pap2-2006}
Papageorgiou V.G., Tongas A.G., Veselov A.P., Yang--{B}axter maps and
 symmetries of integrable equations on quad-graphs, \href{https://doi.org/10.1063/1.2227641}{\textit{J.~Math. Phys.}}
 \textbf{47} (2006), 083502, 16~pages, \href{https://arxiv.org/abs/math.QA/0605206}{arXiv:math.QA/0605206}.

\bibitem{QRT:1988}
Quispel G.R.W., Roberts J.A.G., Thompson C.J., Integrable mappings and soliton
 equations, \href{https://doi.org/10.1016/0375-9601(88)90803-1}{\textit{Phys. Lett.~A}} \textbf{126} (1988), 419--421.

\bibitem{QRT:2006}
Roberts J.A.G., Quispel G.R.W., Creating and relating three-dimensional
 integrable maps, \href{https://doi.org/10.1088/0305-4470/39/42/L03}{\textit{J.~Phys.~A: Math. Gen.}} \textbf{39} (2006),
 L605--L615.

\bibitem{Sakai}
Sakai H., Rational surfaces associated with affine root systems and geometry of
 the {P}ainlev\'{e} equations, \href{https://doi.org/10.1007/s002200100446}{\textit{Comm. Math. Phys.}} \textbf{220} (2001),
 165--229.

\bibitem{Sergeev:1998}
Sergeev S.M., Solutions of the functional tetrahedron equation connected with
 the local {Y}ang--{B}axter equation for the ferro-electric condition,
 \href{https://doi.org/10.1023/A:1007483621814}{\textit{Lett. Math. Phys.}} \textbf{45} (1998), 113--119,
 \href{https://arxiv.org/abs/solv-int/9709006}{arXiv:solv-int/9709006}.

\bibitem{sklyanin-1988}
Sklyanin E.K., Classical limits of {${\rm SU}(2)$}-invariant solutions of the
 {Y}ang--{B}axter equation, \href{https://doi.org/10.1007/BF01084941}{\textit{J.~Sov. Math.}} \textbf{40} (1988),
 93--107.

\bibitem{Veselov:2003b}
Suris Yu.B., Veselov A.P., Lax matrices for {Y}ang--{B}axter maps,
 \href{https://doi.org/10.2991/jnmp.2003.10.s2.18}{\textit{J.~Nonlinear Math. Phys.}} \textbf{10} (2003), suppl.~2, 223--230,
 \href{https://arxiv.org/abs/math.QA/0304122}{arXiv:math.QA/0304122}.

\bibitem{Tsuda:2004}
Tsuda T., Integrable mappings via rational elliptic surfaces,
 \href{https://doi.org/10.1088/0305-4470/37/7/014}{\textit{J.~Phys.~A: Math. Gen.}} \textbf{37} (2004), 2721--2730.

\bibitem{Veselov:1991}
Veselov A.P., Integrable maps, \href{https://doi.org/10.1070/RM1991v046n05ABEH002856}{\textit{Russian Math. Surveys}} \textbf{46}
 (1991), no.~5, 1--51.

\bibitem{Veselov:2003}
Veselov A.P., Yang--{B}axter maps and integrable dynamics, \href{https://doi.org/10.1016/S0375-9601(03)00915-0}{\textit{Phys.
 Lett.~A}} \textbf{314} (2003), 214--221, \href{https://arxiv.org/abs/math.QA/0205335}{arXiv:math.QA/0205335}.

\bibitem{Veselov:2007}
Veselov A.P., Yang--{B}axter maps: dynamical point of view, in Combinatorial
 Aspect of Integrable Systems, \textit{MSJ Mem.}, Vol.~17, Math. Soc. Japan,
 Tokyo, 2007, 145--167.

\bibitem{Veselov:1993}
Veselov A.P., Shabat A.B., Dressing chains and the spectral theory of the
 Schr\"odinger operatorr, \href{https://doi.org/10.1007/BF01085979}{\textit{Funct. Anal. Appl.}} \textbf{27} (1993),
 81--96.

\bibitem{Viallet:2009}
Viallet C.M., Integrable lattice maps: {$Q_{\rm V}$}, a rational version of {$Q_4$},
 \href{https://doi.org/10.1017/S0017089508004874}{\textit{Glasg. Math.~J.}} \textbf{51} (2009), 157--163, \href{https://arxiv.org/abs/0802.0294}{arXiv:0802.0294}.

\bibitem{yang-1967}
Yang C.N., Some exact results for the many-body problem in one dimension with
 repulsive delta-function interaction, \href{https://doi.org/10.1103/PhysRevLett.19.1312}{\textit{Phys. Rev. Lett.}} \textbf{19}
 (1967), 1312--1315.

\end{thebibliography}
\end{document}